\documentclass[twocolumn,nospthms]{svjour3}          
\smartqed  
\usepackage[pdftex]{graphicx}
\usepackage{amsmath}
\usepackage{amssymb}
\usepackage{amsthm}
\usepackage{fullpage}
\usepackage[utf8]{inputenc}
\usepackage[T1]{fontenc}
\usepackage{hyperref}
\usepackage{todonotes}

\newtheorem{lemma}{Lemma}
\newtheorem{theorem}{Theorem}
\newtheorem{definition}{Definition}
\newtheorem{conjecture}{Conjecture}

\newcommand{\mK}{\mathcal{K}}
\newcommand{\mB}{\mathcal{B}}
\newcommand{\mP}{\mathcal{P}}
\newcommand{\mR}{\mathbb{R}}

\newcount\colveccount
\newcommand*\colvec[1]{
        \global\colveccount#1
        \begin{pmatrix}
        \colvecnext
}
\def\colvecnext#1{
        #1
        \global\advance\colveccount-1
        \ifnum\colveccount>0
                \\
                \expandafter\colvecnext
        \else
                \end{pmatrix}
        \fi
}

\usepackage{epstopdf}
\usepackage{subcaption}

\begin{document}

\title{Zeno chattering of rigid bodies with multiple point contacts}

\author{Tam\'as Baranyai \and P\'eter L. V\'arkonyi}
\institute{T. Baranyai \at
              Department of Mechanics, Materials and Structures, Budapest University of Technology and Economics, Budapest, Hungary \\
              \email{baranyai@szt.bme.hu}           
           \and
           P. L. Várkonyi \at
             Department of Mechanics, Materials and Structures, Budapest University of Technology and Economics, Budapest, Hungary \\
              \email{vpeter@mit.bme.hu}
}
\date{2017.04.10}

\maketitle
\begin{abstract}Ideally rigid objects establish sustained contact with one another via complete chatter (a.k.a. Zeno behavior), i.e. an infinite sequence of collisions accumulating in finite time. Alternatively, such systems may also exhibit a finite sequence of collisions followed by separation (sometimes called incomplete chatter). Earlier works concerning the chattering of slender rods in two dimensions determined the exact range of model parameters, where complete chatter is possible. We revisit and slightly extend these results. Then the bulk of the paper examines the chattering of three-dimensional objects with multiple points hitting an immobile plane almost simultaneously. In contrast to rods, the motion of these systems is complex, nonlinear, and sensitive to initial conditions and model parameters due to the possibility of various impact sequences. These difficulties explain why we model this phenomenon as a nondeterministic discrete dynamical system. We simplify the analysis by assuming linearized kinematics, frictionless interaction, by neglecting the effect of external forces, and by investigating objects with rotational symmetry. Application and extension of the theory of common invariant cones of multiple linear operators enable us to find sufficient conditions of the existence of initial conditions, which give rise to complete chatter. Additional analytical and numerical investigations predict that our sufficient conditions are indeed exact, moreover solving a simple eigenvalue problem appears to be enough to judge the possibility of complete chatter. 
\end{abstract}

%

\section{Introduction}
Ideally rigid objects establish sustained contact with one another via an infinite sequence of collisions accumulating in finite time. This phenomenon is commonly referred to as complete chatter \cite{nordmark2009simulation} or Zeno behaviour \cite{ames2006there}, we will adopt the first one of these two names. Complete chatter can be observed in various situations like that of a ball bouncing on a horizontal surface \cite{luck1993bouncing}, in the extensively studied problem of rocking blocks \cite{Housner,zhang2001rocking}, the motion of Newton's craddle \cite{donahue2008newton,ceanga2001new} and of Euler's Disk immediately before reaching its final singularity \cite{moffatt2000euler,le2005dynamics}. This phenomenon has important implications to the Lyapunov stability analysis of rigid bodies with unilateral contacts \cite{varkonyi2016lyapunov}. Complete chatter also plays a significant role in the theory of hybrid dynamical systems \cite{shen2005linear,goebel2008zeno,ames2007sufficient} and in optimal control problems \cite{borisov2000fuller}.

The simplest example of the bouncing ball reveals that an impact model, and a model of continuous dynamics are both necessary to analyze chattering. The prototypical example of a falling rod additionally shows that the chattering of objects with multiple potential impact locations (e.g. the two endpoints of the rod) is a complex phenomenon because the locations of subsequent impacts on the object may follow many different periodic or possibly chaotic patterns.

The falling rod problem was first studied in detail by Goyal et. al. \cite{goyal1,goyal2}. Their model was simplified via linearisation of the rotational kinematics during impact-free motion, and by neglection of the effect of external forces (including gravity). The second assumption is too crude for the analysis of the bouncing ball, which would leave the surface immediately after the first impact without the effect of gravity. Nevertheless it is plausible in the case of the rod problem, if the two endpoints of the rod hit the floor nearly simultaneously, i.e. if the time intervals between subsequent impacts are very short. Or \cite{or2007frictional}  found that this assumption yields  correct results for most combinations of physical parameters.

 Goyal et. al. studied the number of collisions before the rod leaves the surface as a function of system parameters and initial conditions. Their main result was to identify two distinct regimes in parameter space. In one of these regimes, the rod leaves the surface after a finite number of collisions, whereas in the other, it may undergo complete chatter, after which the rod stops in contact with the floor. By varying the parameters of the sytem, a sudden transition between these two, qualitatively different behaviours can be observed.

In this paper, we extend the analysis of \cite{goyal1,goyal2} to objects with more than two potential impact locations. We will study in detail the motion of three-dimensional objects hitting a flat surface, under the assumption that the  contact points form a regular $n$-gon or the affine image of a regular $n$-gon. We develop several sufficient conditions of the possibilty of complete chatter. By using extensive numerical simulation, and semi-analytic investigation of the impact maps, we form several conjectures with respect to the exact conditions of the existence of complete chatter. Some of these are proven for $n=4$ (squares, rectangles, parallelograms) semi-analytically.

We use more complex mathematical tools than \cite{goyal1,goyal2} whose assumptions ensure that the two endpoints of the rod may not hit the ground but in alternating order. In contrast, the objects considered here have more than two impact points, which may hit the ground in many possible orders.

Our investigations are tightly related to an extensive body of work related to Zeno behavior in hybrid systems. These works establish sufficient conditions of the existence \cite{goebel2008zeno,ames2007sufficient} or the non-existence\cite{shen2005linear,goebel2008zeno} of complete chatter in various systems as well as conditions of local attractivity of so-called "Zeno equilibrium states" \cite{goebel2008lyapunov,lamperski2013lyapunov,or2011stability,murti2014using}  (i.e. those states where complete chattering sequences terminate). Our main results are also sufficient conditions of the existence of Zeno behavior. Nevertheless previous works are typically limited to hybrid systems displaying regular, periodic sequences of discrete state transitions and thus they are not applicable to the problem of falling objects, where complex sequences of impacts are possible.

The rest of the paper is organized as follows.
In Sec. \ref{sec:mech}, we introduce the system to be investigated, an impact model, and our basic notation. 
In Sec. \ref{sec:rod}, the rod problem is revisited and the main results of \cite{goyal1,goyal2} are reproduced with the help of the classical theory of invariant cones of linear operators (or Perron-Frobenius theory) \cite{vandergraft}. Our approach to three-dimensional objects is based on an extended version of the theory dedicated to shared invariant cones of multiple operators. The general theory has been developed recently \cite{prot} and has found several other applications in the stability analysis of switching dynamical systems \cite{shen2015stability,ogura2013stability} and in the control theory of linear systems \cite{do2006invariant,angeli2003monotone}. Nevertheless our application requires further generalization of the theory, which is introduced in Sec. \ref{sec:4}. 
The main results of the paper are presented in Sec. \ref{sec:effinvcone}, which includes 
semi-analytic sufficient conditions of complete chatter in the case of squares (Sec. \ref{sec:square}) and rectangles (Sec. \ref{sec:affine}) as well as systematic numerical simulations (Sec. \ref{sec:num}), which strongly suggest that our sufficient condition is exact, and it also applies to any regular polygon. The paper is closed by a Discussion section.

\section{The mechanical problem}\label{sec:mech}

\subsection{Problem statement}
We are interested in the chattering motion of a rigid body $\mathcal{B}$ i.e. the rapid sequence of collisions occuring when it hits a flat plane $\mathcal{P}$. The object is assumed to have a finite number of coplanar vertices forming a convex polygon, which may potentially contact the surface. 

Chattering is a complex, hybrid non-smooth and non-linear motion composed of periods of smooth dynamics and sudden impacts. To simplify the problem, we focus on the case when the vertices of the object reach the surface almost simultaneously. This assumption will allow us to use linearised kinematics (Sec. \ref{sec:linearized}). Furthermore, we assume that gravity and other external forces (except for impulsive contact forces) have negligible time to act, thus their effect will be ignored. 

Our goal is to predict - for given initial conditions, shape and mechanical properties of the body - which one of the following two qualitatively  different behaviours occurs :
\begin{itemize}
\item \textbf{incomplete chatter (ICC)}: the body leaves the surface with finite velocity after finite number of collisions
\item \textbf{complete chatter (CC)}: the object undergoes an infinite sequence of collisions in finite time, after which its velocity relative to the surface surface becomes zero. 
\end{itemize}
As we will see, it is necessary to consider a special scenario as well:
\begin{itemize}
\item \textbf{partial complete chatter (PCC)}: two vertices of an object with $n>2$ vertices undergo an infinite sequence of collisions in finite time, after which the velocity of these two points relative to the surface becomes zero but the object stays in motion. The PCC sequence may start at the beginning of the motion or after an initial transient.
\end{itemize}
After a PCC sequence, if other vertices of the object move away from the surface, then no more impacts will occur, thus this scenario is indeed similar to an ICC. It is also possible that other vertices of the object move towards the surface after the PCC sequence, which will eventually lead to a simultaneous collision of all vertices with the surface. Because of the notorious difficulty of modelling simultaneous impacts, we will categorize this case as undecidable.

\subsection{Notation and kinematics}

Let the mass of the object $\mathcal{B}$ and its mass moment of inertia tensor be $m$ and $\vec{\theta}$. We define a local orthogonal coordinate system fixed to it's centre of mass $\vec{r}$.   The unit vectors spanning the local frame are denoted by $\vec{u}_x$, $\vec{u}_y$ and $\vec{u}_z$. The x-y plane of the reference frame is parallel to the plane spanned by $n$ potential contact points of the object. Physical quantities expressed in local frame will be denoted by upper indices $l$. 

The coordinates of the contact points  are
\begin{align}
\vec{\vec{r}}^l_i=\colvec{3}{x_i}{y_i}{z_*}, i\in \{0,1...n-1\} \label{eq:koord}
\end{align} 
where the third coordinate $z_*$ is identical for all vertices. The enumeration of the vertices ($i=0,1,...,n-1$) is according to positive orientation in the local frame.

We will assume that the axes of the local frame correspond to the eigenvectors of $\vec{\theta}$, i.e.
$$
\vec{\theta}^l=m
\begin{bmatrix}
\rho_x^2 & 0 & 0\\
0 & \rho_y^2 &0\\
0 & 0 & \rho_z^2
\end{bmatrix}
$$
where $\rho_x,\rho_y,\rho_z$ denote the principal radii of gyration of $\mathcal{B}$.

We also consider a global orthogonal frame of reference, whose $X$ and $Y$ axes are parallel to the flat plane $\mathcal{P}$ and whose origin is at distance $z_*$ from $\mP$. We also define the unit vectors $\vec{u}_X$, $\vec{u}_Y$, $\vec{u}_Z$ spanning the global frame and use upper index $g$ for quantities expressed in global frame.

The position of the $i$-th vertex in global frame is given by 
\begin{align}
\vec{r}_i^g=\vec{r}^g+\vec{H}_{l,g} \vec{r}_i^l. \label{eq:pontoshelyzet} 
\end{align}   
where $\vec{H}_{l,g}$ is a rotation matrix. The velocity $\vec{v}_i$ of point $i$ is
\begin{align}
\vec{v}_i^g=\vec{v}^g+\vec{\omega}^g \times \vec{H}_{l,g} \vec{r}_i^l. \label{eq:pontossebesseg} 
\end{align}  
where $\vec{\omega}$ is the angular velocity of $\mathcal{B}$ and $\vec{v}$ is the velocity of the center of mass

\subsection{Continuous dynamics}
Between two impacts, the smooth dynamics of the object is given by the Newton-Euler equations. The time-derivatives of the velocity of the center of mass $\vec{v}$ and the angular velocity $\vec{\omega}$ are determined by the external forces, which are of size $O(1)$. We are interested in a rapid sequence of collisions, which means that variations of $\vec{v}$ and $\vec{\omega}$ between two collisions are very small. Thus they are approximated by constants. 

\subsection{Impacts}\label{sec:impacts}
The object undergoes an impact if one of the vertices hits the plane, i.e. if 
\begin{align}
\vec{u}_Z^T\vec{r}_i=0
\quad
\vec{u}_Z^T\vec{v}_i<0
\end{align}
for some $i$.

For simplicity, we assume zero friction implying that the impact impulses are parallel to the contact normal $\vec{u}_Z$. We remark without detailed proof that all results of the paper remain valid in the presence of friction, provided that the object is flat, i.e. $z_*=0$ in \eqref{eq:koord}. 

Let $F\cdot \vec{u}_Z$ ($F\in \mathbb{R}$) denote an instantaneous impulse the underlying plane excerts upon $\mathcal{B}$ in a single-point impact. The pre- and post-impact values of the velocity of the centre of mass, the angular velocity and the velocity of vertex $i$ will be distinguished by superscripts $^-$ and $^+$. The conservation of linear and angular momenta yields  
  
\begin{align}
m (\vec{v}^{+}-\vec{v}^{-})=F\vec{u}_Z \label{eq:utk1}\\
\vec{\theta} (\vec{\omega}^+-\vec{\omega}^-) = \vec{r}_i\times F\vec{u}_Z \label{eq:utk2}
\end{align}
We assume partially elastic collisions with a constant Newtonian coefficient of restitution $0<\gamma<1$, implying
 \begin{align}
 \vec{u}_Z^T \vec{v}_i^+ = -\gamma \vec{u}_Z^T \vec{v}_i^- \label{eq:utk3}
 \end{align}
The unknowns $\vec{v}^+$, $\vec{\omega}^+$, and $F$ are uniquely determined by the equations \eqref{eq:utk1}-\eqref{eq:utk3} as follows:
 \begin{align} 
F&=\frac{-\vec{u}_z^T(\gamma+1)(\vec{v}^-+(\vec{\omega}^- \times \vec{r}_i))}{\vec{u}_z^T\left[m^{-1} \vec{I} - \vec{R}_x \vec{\theta}^{-1} \vec{R}_x \right] \vec{u}_z } \label{eq:F} \\
\vec{\omega}^+&=\vec{\theta}^{-1}(\vec{r}_i^g \times F \vec{u}_z)+\vec{\omega}^-  \label{eq:omega+}\\
\vec{v}^+&=m^{-1}F \vec{u}_z +\vec{v}^- \label{eq:v+}
\end{align}
where $\vec{I}$ stands for the identity matrix in $3$ dimensions, $\vec{r}_i^g=\vec{H}_{l,g} \vec{r}_i^l$ and $\vec{R}_x$ is the matrix representation of the cross product $\vec{r}_i^g\times*$ (i.e. $\vec{R}_x \vec{x}=\vec{r}_i^g\times \vec{x}$ for all $\vec{x} \in \mathbb{R}^3$).

We have pointed out that a PCC event with other vertices moving towards the surface $\mathcal{S}$ gives rise to a simultaneous impact at all vertices. There are various simple models of simultaneous impacts, which yield reasonable but not reliable results. The lack of reliability is caused primarily by the extreme sensistivity of simultaneous impacts to the pre-impact state \cite{brogliato1999nonsmooth},\cite{stronge2004impact}. 
One common assumption is that a simultaneous impact can be replaced by a (possibly infinite) sequence of single-point impacts \cite{chatterjee1998new}\cite{varkonyi2012lyapunov},  whereas another popular approach is to assign a coefficient of restitution to all vertices and formulate the impact model as a linear complementarity problem \cite{glocker1995multiple} \cite{leine2008stability}. We have found that in the case of chattering, these two models predict qualitatively different results (the object stops under the first assumption and it topples under the second), which motivates our decision to categorize the final outcome of motion including a PCC event  as undecidable.

\subsection{Linearisation and general coordinates}
\label{sec:linearized}

If the vertices of the object reach the surface nearly simultaneously, then its rotations during the whole chattering process remain very small. Thus we obtain a good approximation of the motion via application of the theory of infinitesimal rotations. Small rotations can be represented by a rotation vector $\vec{\phi}\in\mR ^3$ where the direction of the vector represents the axis of the rotation and  $|\vec{\phi}|<<1$ corresponds to the angle of the rotation. This vector is related to angular velocity as $\vec{\omega}=\frac{d}{d\tau}\vec{\phi}+\vec{O}(|\vec{\phi}|^2)$ ($\tau$ stands for time). Such a rotation is equivalent of the rotation matrix 
\begin{align}
\vec{H}_{l,g}=
\begin{bmatrix}
1 & -\vec{u}_Z^T\vec{\phi} & \vec{u}_Y^T\vec{\phi}\\
\vec{u}_Z^T\vec{\phi} & 1 & -\vec{u}_X^T\vec{\phi}\\
-\vec{u}_Y^T\vec{\phi}  & \vec{u}_X^T\vec{\phi} & 1
\end{bmatrix}
+\vec{O}(|\vec{\phi}|^2)
\label{eq:Hlin}
\end{align} 
Throughout the paper, we neglect $\vec{O}(|\vec{\phi}|^2)$ terms. In addition, we assume that the object reaches the surface without "yaw motion", i.e. $\vec{u}_Z^T\vec{\omega}=0$ implying $\vec{u}_Z^T\vec{\phi}=constant$. Without loss of generality, we will assume $\vec{u}_Z^T\vec{\phi}=0$. Then, \eqref{eq:pontoshelyzet} and \eqref{eq:Hlin} imply
\begin{align}
\vec{r}_i^g=\colvec{3}
{x_i}
{y_i}
{\vec{u}_Z^T\vec{r} + \vec{u}_X^T\vec{\phi} y_i  - \vec{u}_Y^T\vec{\phi} x_i} 
\end{align}
i.e. the distance of a vertex from $\mathcal{P}$ is determined as a linear combination of  $h:= \vec{u}_Z^T\vec{r}$, $\phi_x:= \vec{u}_X^T\vec{\phi}$, and $\phi_y:= \vec{u}_Y^T\vec{\phi}$. This motivates our choice of the generalized coordinates  
\begin{align*}
\vec{q}=\colvec{3}{\phi_x}{\phi_y}{h}
\end{align*}
 spanning the configuration space $\mathbb{C}$. The velocity space $\mathbb{V}$ is spanned by the generalized velocities 
 \begin{align*}
 \vec{p}=\frac{d\vec{q}}{d\tau}=\colvec{3}{\omega_x}{\omega_y}{v}
 \end{align*}

Let us introduce the notation
\begin{align}
\vec{f}_i:=\colvec{3}{y_i}{-x_i}{1}
\label{eq:fi}
\end{align}

Then the distance of a vertex from $\mP$ and its velocity in the global $Z$ direction can be expressed as 
\begin{align}
h_i=\vec{u}_z^T \vec{r}_i=\vec{f}_i^T \vec{q}\;\;\; v_{i}=\vec{u}_z^T \vec{v}_{i}=\vec{f}_i^T\vec{p} \label{eq:hi,vi}
\end{align}
 According to \eqref{eq:omega+}, \eqref{eq:v+}, an impact at vertex $i$ corresponds to a linear mapping of the generalized velocities
\begin{align}
\vec{p}^+=\vec{U}_i\vec{p}^-\label{eq:impactmap}
\end{align} 
where 
\begin{align}
\vec{U}_i:=\vec{I}+\frac{-(1+\gamma)}{\vec{f}_i^T \vec{\Theta}^{-1}\vec{f}_i}\left[\vec{\Theta}^{-1}\vec{f}_i \vec{f}_i^T \right]\label{eq:impactequ}
\end{align} 
furthermore $\vec{\Theta}$ is a generalized inertia matrix: 
\begin{align}
\vec{\Theta}=
  m\begin{bmatrix}
    \rho_x^2 & 0 & 0 \\
    0 & \rho_y^2 & 0 \\
    0 & 0 & 1
  \end{bmatrix}
  \label{eq:Theta}
\end{align}
  and   $\rho_x$ and $\rho_y$ are radii of gyration introduced earlier.

\subsection{Invariant cones and complete chatter}
If the center of mass of $\mathcal{B}$ approaches $\mathcal{P}$, i.e. if 
\begin{align}
\vec{u}_3^T\cdot \vec{p}\leq 0\label{eq:vz<0}\\
\vec{u}_3=[0\;0\;1]^T\label{eq:u3}
\end{align} 
during the entire motion, then $\mathcal{B}$ must eventually undergo CC or PCC. In contrast, because the contact forces are unilateral, the normal velocity of the center of mass increases monotonically during the motion, hence if $[0 \; 0 \; 1]\cdot \vec{p}>0$ at any time, then $\mathcal{B}$ will never become immobile, implying an ICC. This observation is summarized in

\begin{lemma} An object undergoes ICC if and only if \eqref{eq:vz<0} is violated at any time during its motion
\label{lemma:vz<0}
\end{lemma}

 The only events at which $\vec{p}$ changes are the impacts, which are modelled by the linear maps \eqref{eq:impactmap}. Hence, we will investigate whether or not \eqref{eq:vz<0} is preserved, and the main tool of our investigation is the invariant cone theory of linear maps, which is briefly reviewed below.

Consider a vector space, such as the velocity space $\mathbb{V}$. The set $\mathcal{K} \subset \mathbb{V}$ is a \emph{cone} if $ \vec{x} \in \mathcal{K}$ implies $\alpha \vec{x} \in \mathcal{K}$ for any $\alpha \geq 0$. $\mathcal{K}$ is a \textit{proper cone} if it is convex and non-empty ($\mathcal{K}+\mathcal{K} \subseteq \mathcal{K}$)and also pointed ($\mathcal{K}\cap -\mathcal{K}=0$). Consider now a linear operator $\vec{A}:\mathbb{V}\rightarrow \mathbb{V}$! Then,
\begin{definition}the cone $\mathcal{K}$ is called an invariant cone of $\vec{A}$  if $\vec{A}(\mathcal{K})\subseteq\mathcal{K}$.
\end{definition} 
Whether or not a given operator has an invariant proper cone or not is decidable with the aid of
\begin{theorem} [Elsner-Vandergraft \cite{vandergraft}] 
Let $\{\lambda_i\}$ denote the set of eigenvalues of $\vec{A}$, and let $\{\lambda_i^{dom}\} \subseteq \{\lambda_i\}$ denote the set of dominant eigenvalues, i.e. those eigenvalues for which $|\lambda_i|=\max_i|\lambda_i|$. Then  $\vec{A}$ has an invariant proper cone if and only if there exists a dominant eigenvalue $\lambda_p \in \{\lambda_i^{dom}\}$ such that\\
(i)  $\lambda_p$ is real and positive\\
(ii) the algebraic multiplicity of $\lambda_p$ is not less than the multiplicity of any other $\lambda_i \in \{\lambda_i^{dom}\}$.\\
Furthermore, if (i) and (ii) are satisfied then any invariant cone must contain an eigenvector corresponding to $\lambda_p$. 
\label{thm:invconetheorem}
\end{theorem}
We will see in Sec. 3 that in the case of a rod-shaped object with only two potential contact points, Theorem 1 leads to an exact condition of CC. Nevertheless, the object $\mathcal{B}$ has in general more than two possible contact points, i.e. several impact operators. It is a sufficient condition of CC  that the property \eqref{eq:vz<0} is preserved by an \textit{arbitrary sequence} of impacts, for which we need the more general concept of common invariant cones:
\begin{definition}
A cone $\mathcal{K}$ is called a common invariant cone for a set of linear operators  $\mathcal{A}=\{\vec{A}_1...\vec{A}_n \}$ if $\vec{A}_i(\mathcal{K}) \subseteq\mathcal{K} \ \forall \vec{A}_i \in \mathcal{A}$.
\end{definition}
The theory of common invariant cones has not been investigated until very recently. There is provably no efficient general algorithm to decide the existence of a common invariant cone for an arbitrary set of operators \cite{prot,jointspectral}. Nevertheless, it is possible in certain cases to prove its existence by construction. 

In our problem, stronger results can be achieved by taking into account that for a given value of $\vec{p}$, certain impacts may be impossible. More specifically, we will identify additional cones $\mathcal{C}_i$ in velocity space such that an impact at vertex $i$ is impossible unless $\vec{p}\in \mathcal{C}_i$. We also define a new concept:
\begin{definition}
A cone  $\mathcal{K}$ is called an effectively invariant cone for a set of operators $\mathcal{A}=\{\vec{A}_1...\vec{A}_n \}$ and a set of cones $\mathcal{C}=\{C_1...C_n \}$ if $\vec{A}_i(\mathcal{K} \cap \mathcal{C}_{i}) \subseteq \mathcal{K} $ for all $i\in \{1...n\}.$
\end{definition}
There is no available method for testing the existence of an effectively invariant cone, nevertheless we will be able to prove their existence by construction in certain cases.

\section{The falling rod revisited}\label{sec:rod}

\begin{figure}[h]
\begin{center}
\includegraphics[width=70mm]{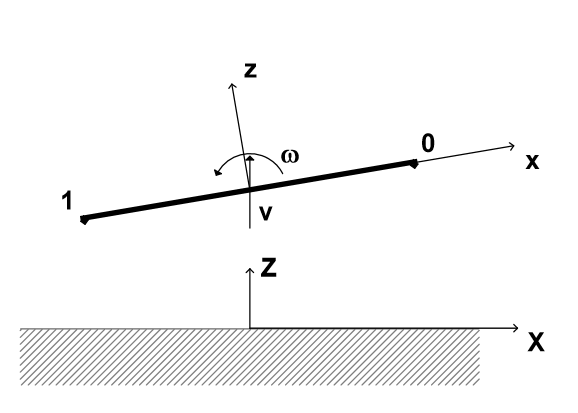}
\caption{Notation of the rod problem}
\label{fig:rod}
\end{center}
\end{figure}
As we have seen, a rigid body needs at least two potential points of collision to display chattering, as long as the effect of external forces is neglected. Having exactly two points considerably simplifies the problem, because the order of colliding nodes becomes trivial. The point that has undergone a collision moves upwards, and only the other one may hit the surface, thus collisions occur alternating at the two vertices. Goyal et al. \cite{goyal1,goyal2} investigated this problem, and gave the exact conditions of complete chatter if the rod initially performs pure translational motion.

In contrast to the general problem investigated in the paper, we now consider planar motion   (Fig. \ref{fig:rod}). The position vectors of the endpoints are two-dimensional: $\vec{r}_i^l=[x_i \ 0]^T$ ($i=0,1$). The rotation angle $\phi$ and the angular velocity $\omega$ are scalars and we have a scalar radius of inertia $\rho$.  Accordingly, we will use the two dimensional generalized coordinate and velocity vectors $\vec{q}=[\phi\quad h]^T$, and $\vec{p}=[\omega \quad v]^T$. The relations \eqref{eq:hi,vi}-\eqref{eq:impactequ} remain true but \eqref{eq:fi}  and \eqref{eq:Theta} are replaced by
\begin{align}
\vec{f}_i=\begin{bmatrix}
x_i\\ 1\end{bmatrix}, \quad
\vec{\Theta}=m
\begin{bmatrix}
\rho & 0\\0&1
\end{bmatrix}
\end{align}
 In this situation, an immediate analogue of Lemma \ref{lemma:vz<0} holds:
\begin{lemma} A rod in planar motion undergoes ICC if and only if 
\begin{equation}
[0\,\,\,1]\vec{p}\leq 0
\label{eq:felfele}
\end{equation}
 is violated at any time during its motion
\label{lemma:vz<0planar}
\end{lemma}
We will assume that the rod has a symmetric mass distribution, and $x_0=-x_1=1$. Temporarily, we will also assume that the first collision occurs at $\vec{r}_1$. This assumption means that the initial velocity of point 1 must point downwards, yielding the constraint
\begin{align}
[-1\;1]\cdot \vec{p}^{(0)}<0
\label{eq:r1down}
\end{align}
where $\vec{p}^{(0)}$ is the iniitial value of $p$.

In order to simplify our analysis, we will change the reference frames and also swap the labelling of the two endpoints before every impact (including the first one). As a result, every impact will occur at the endpoint labelled by 0 at the time of the impact. Specifically, the directions of the $X$ and the $x$ coordinate axes are both reversed every time, which corresponds to the transformations
\begin{align}
\vec{q}&\rightarrow \vec{P} \vec{q}\label{eq:Pq}\\
\vec{p}&\rightarrow \vec{P} \vec{p}\label{eq:Pp}\\
\vec{r}_i^l&\rightarrow \vec{P} \vec{r}_i^l \label{eq:Pri}
\end{align}
with
\begin{align}
\vec{P}=
\left[
\begin{matrix}
-1&0\\0&1
\end{matrix}
\right] 
\end{align}
whereas all other system parameters and equations governing the dynamics of the rod remain unchanged. Second, swapping the labels 0 and 1 corresponds to a second transformation of the local coordinates $\vec{r}_i^l$ identical to \eqref{eq:Pri}. The two steps leave $\vec{r}_i^l$ unchanged (due to $\vec{PP}=identity$), whereas the combined effect of the technical steps and the subsequent impact to the generalized velocity is a linear transformation
\begin{align}
\vec{p}\rightarrow \vec{U}_0\vec{P}\vec{p}
\end{align} 
If the technical steps are repeated before every impact, then the generalized velocity $\vec{p}^{(k)}$ of the rod after $k$ impacts will be 
\begin{align}
\vec{p}^{(k)}&=(\vec{U}_0\vec{P})^k \vec{p}^{(0)} \label{eq:pk1rod}
\end{align}

If - in contrast to our initial assumption - point 0 hits the ground first, then the coordinate transformation and relabelling step before the first impact are omitted, and \eqref{eq:pk1rod} is replaced by 
\begin{align}
\vec{p}^{k}=(\vec{U}_0\vec{P})^{k-1}\vec{U}_0 \vec{p}^{(0)}=(\vec{U}_0\vec{P})^{k} \vec{P}\vec{p}^{(0)}
\label{eq:pk0rod}
\end{align}
From the last two expressions, it is clear that whether or not CC occurs depends mostly on properties of the matrix $\vec{U}_0 \vec{P}$. Indeed, \cite{goyal1,goyal2} showed 
\begin{theorem}
Let the initial motion of the rod be pure translation towards the support surface (i.e.: $\vec{p}^{(0)}=[0\quad v^{(0)}]^T$ with $v^{(0)}<0$). If
\begin{align}
\rho \leq \frac{2}{\sqrt{\gamma} + 1}-1
\label{eq:CCcond}
\end{align}
then CC occurs. In the converse case, ICC occurs.
\label{thm:goyal}
\end{theorem}
We can reproduce the original proof of Theorem \ref{thm:goyal} by using invariant cones as follows. 
\begin{proof}
The matrix $\vec{U}_0\vec{P}$ can be expressed as
\begin{align}
\vec{U_0P}=\begin{bmatrix}
\frac{\gamma+1}{\rho^2+1}-1 & -\frac{\gamma+1}{\rho^2+1}\\
\rho^2 \frac{\gamma+1}{\rho^2+1} & \frac{\gamma+1}{\rho^2+1} - \gamma
\end{bmatrix}
\end{align}
Its eigenvalues and eigenvectors are
\begin{align}
\lambda_{max},\lambda_{min}=\frac{(1-\rho^2) (\gamma+1)\pm\sigma^{1/2}}{2(\rho^2 + 1)}\\
\vec{p}_{max},\vec{p}_{min}=\left(\frac{(1-\gamma)(\rho^2+1)\mp\sigma^{1/2}}{2 \rho^2( \gamma+1 )},-1 \right)^T
\label{eq:pmaxpmin}
\end{align}
where
\begin{align}
\sigma=\gamma^2 - 2\rho^4 \gamma - 12\rho^2 \gamma - 2 \rho^2 \gamma^2 + \rho^4 \gamma^2 - \rho^2 + \rho^4 -  \gamma + 1
\end{align}
The eigenvalues have the following properties (Fig. \ref{fig:sertekgorbe}):
\begin{enumerate}
\item If \eqref{eq:CCcond} is true, then $\sigma\geq 0$, and thus the eigenvalues are real. Furthermore $\vec{U}_0\vec{P}$ has positive trace and determinant implying that both eigenvalues are positive.
\item Similarly, if \eqref{eq:CCcond} is not satisfied furthermore
\begin{align}
\rho \geq \frac{-2}{\sqrt{\gamma} - 1}-1 \label{eq:NCcond}
\end{align}
then the eigenvalues are real and negative,
\item If none of the two conditions listed above are satisfied then they are complex.
\end{enumerate}
\begin{figure}[h]
\begin{center}
\includegraphics[width=84mm]{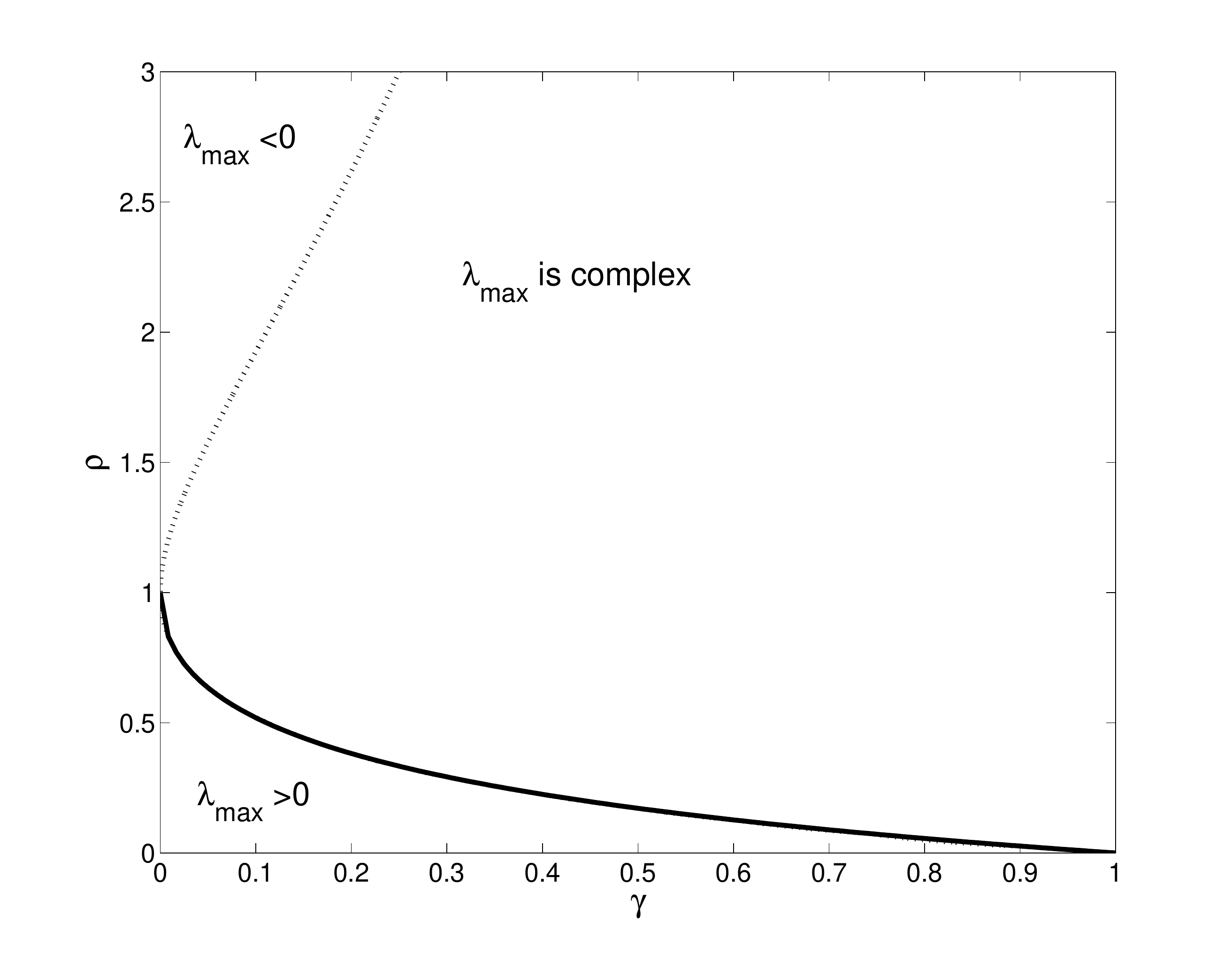}
\caption{Properties of $\lambda_{max}$. The solid and dotted lines are given by \eqref{eq:CCcond} and \eqref{eq:NCcond}}
\label{fig:sertekgorbe}
\end{center}
\end{figure}
\begin{figure}[h]
\begin{center}
\includegraphics[width=84mm]{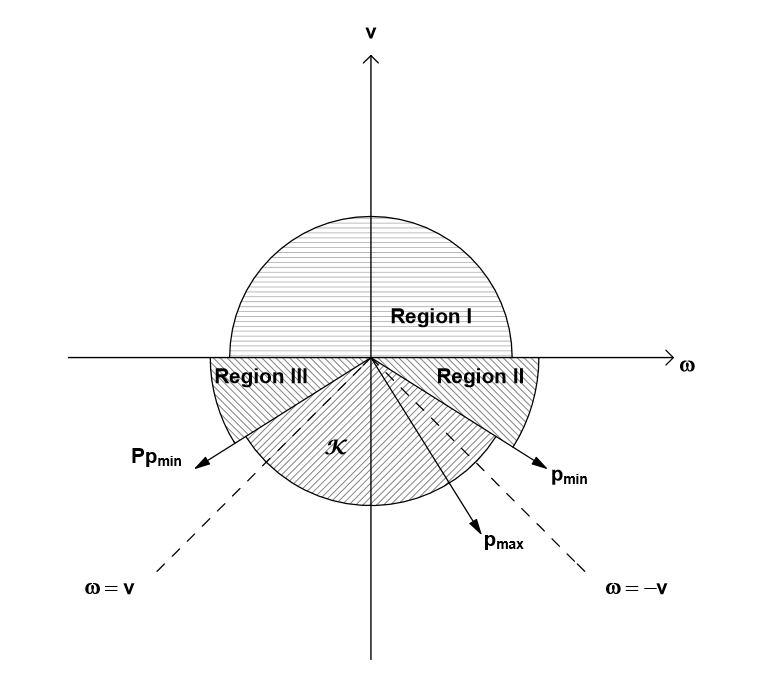}
\caption{Case 1 of the rod problem: the cone $\mathcal{K}$ generated by vectors $\vec{p}_{min}$ and $\vec{P}\vec{p}_{min}$ contains $\vec{p}_{max}$ and $\vec{p}_{trans}$ and is invariant to $\vec{U}_0\vec{P}$.}
\label{fig:Pak}
\end{center}
\end{figure}
In case 1, \eqref{eq:CCcond} and  \eqref{eq:pmaxpmin} yield
\begin{align}
\begin{split}
[1\quad 0] \vec{p}_{min} &=
\frac{(1-\gamma)(\rho^{-2}+1)+\rho^{-2}\sigma^{1/2}}{2(1+\gamma)}\geq\\
&\geq \frac{(1-\gamma)(\rho^{-2}+1)}{2(1+\gamma)}\\
&\geq \frac{1+\gamma^{1/2}}{1-\gamma^{1/2}}\geq 1
\end{split}
\end{align}
Hence, $\vec{p}_{min}$ is in the positive-negative quadrant of velocity space, with the angle between $\vec{p}_{min}$ and $(1, \ 0)^T$ not exceeding $\pi/4$ (Figure \ref{fig:Pak}).
Furthermore $\sigma \geq 0$ in \eqref{eq:pmaxpmin} implies that $\vec{p}_{max}$ is in the cone $\mathcal{K}$ spanned by $\vec{p}_{min}$ and $\vec{P}\vec{p}_{min}$. 

Next, we express  $\vec{P}\vec{p}_{min}$ as
\begin{align}
\vec{P}\vec{p}_{min}=\alpha \vec{p}_{max}-\beta \vec{p}_{min} \quad \alpha,\beta>0
\end{align} 
which means that the generating vectors of $\mK$ are mapped by $\vec{U}_0P$ into
\begin{align}
\begin{split}
\vec{U}_0\vec{P}\cdot \vec{P}\vec{p}_{min}&=\lambda_{max}\cdot \vec{P}\vec{p}_{min}+\\
&\quad +(\lambda_{max}-\lambda_{min})\beta\cdot \vec{p}_{min}
 \in\mathcal{K}
\end{split}\\
\vec{U}_0\vec{P}\vec{p}_{min}&=\lambda_{min}\vec{p}_{min} \in\mathcal{K}
\end{align}
hence $\mK$ is an invariant cone of $\vec{U}_0\vec{P}$.
The initial velocity specified by the theorem satisfies 
$$
p^{(0)}=[0\quad v^{(0)}]^T=\vec{P}[0\quad v^{(0)}]^T\in\mathcal{K}
$$
hence the generalized velocity remains in the invariant cone $\mathcal{K}$ after any number of impacts. Every point in $\mathcal{K}$ satisifes \eqref{eq:felfele}, and thus Lemma \ref{lemma:vz<0planar} implies that CC occurs.

In case 2, the repeated multiplication by $\vec{U}_0\vec{P}$ in \eqref{eq:pk1rod} and \eqref{eq:pk0rod} causes $\vec{p}^{(k)}$ for large values of $k$ to approach the dominant eigenvextor $\vec{p}_{max}$ in the following sense:
\begin{equation}
\lim_{k\rightarrow \infty}\vec{p}^{(k)}\lambda_{max}^{-k}=\alpha \vec{p}_{max}
\label{eq:plim}
\end{equation}
where $\alpha \in \mR$. Since $\lambda_{max}<0$, \eqref{eq:felfele} will be violated either for even or for odd large values of $k$. Thus, Lemma \ref{lemma:vz<0planar} implies ICC.

In case 3, the complex eigenvalues mean that multiplication by $\vec{U}_0\vec{P}$ stretches vectors and rotates them by a constant angle. The rotational component eventually leads to a violation of \eqref{eq:felfele}, thereby Lemma \ref{lemma:vz<0planar}  implies ICC.  
\end{proof}

The works \cite{goyal1,goyal2} did not examine how the completeness of the chatter changes if the initial motion includes a rotational component ($p^{(0)}=[v^{(0)}\;\omega^{(0)}]^T$ with $\omega^{(0)}\neq 0$). We can also use invariant cones to answer this more general question:
\begin{theorem}\label{thm:goy1}
The rod undergoes CC if and only if \eqref{eq:CCcond} is satisfied, and its initial velocity $\vec{p}^{(0)}$ is in the cone $\mathcal{K}$.  
\end{theorem} 

\begin{proof}
The proof of the if part follows from the invariance property of $\mathcal{K}$ in the same manner as explained in the proof of Theorem \ref{thm:goyal}. 
If \eqref{eq:CCcond} is not satisfied, then the proof of the only if part is also the same as in Theorem \ref{thm:goyal}.

Our only remaining task is to prove that ICC occurs if \eqref{eq:CCcond} is satisfied but $\vec{p}^{(0)}\notin\mathcal{K}$. Assume that the first impact occurs at point 1, i.e. \eqref{eq:pk1rod} applies. Then, $\vec{p}^{(0)}$ must be in one of the following regions of Fig. \ref{fig:Pak}:
\begin{itemize}
\item{Region I}: all points of this region violate \eqref{eq:felfele} thus Lemma \ref{lemma:vz<0planar} implies ICC. 
\item{Region II}: here, $\vec{p}^{(0)}$ can be decomposed as 
$\vec{p}^{(0)}=-\alpha \vec{p}_{max} + \beta \vec{p}_{min}$ with $\alpha,\beta>0$. For large values of $k$, we will again have \eqref{eq:plim} but now with $\alpha<0<\lambda_{max}$. The second coordinate of $-\vec{p}_{max}$ is positive, thus Lemma \ref{lemma:vz<0planar} implies ICC.
\item{Region III}: this region is on the left side of $\vec{P}\vec{p}_{min}$, where every point contradicts \eqref{eq:r1down}. Hence $\vec{p}^{(0)}$ cannot be in Region III. 
\end{itemize}
This completes the proof if the first impact occurs at point 1. In the converse case, the roles of Region I and III are interchanged, otherwise the proof remains identical.
\end{proof}

\section{Complete chatter of regular polygons}
\label{sec:4}

\subsection{ A sufficient condition via common invariant cones}\label{sec:commoncone}

We now consider objects with $n\geq 3$ possible impact locations and thus $n$ impact maps. Even though the same point may not hit the ground in two subsequent impacts, this constraint leaves us with infinitely many possible sequences of impact locations, and the approach outlined in the previous section cannot be used. Nevertheless we can obtain a \textit{sufficient} condition of CC or PCC by using the notion of common invariant cones:

As an example, consider an object whose contact points form a regular $n$-gon (Fig \ref{fig:sokszog}) with vertices, 
\begin{align}
\vec{r}_0^l&=[1\;0\;z_*]^T\\
\vec{r}_i^l&=\vec{P}_{2i\pi/n}\vec{r}_0^l\quad \text{for}\quad i=1,2,...,n-1
\end{align}
where $z_*$ is an arbitary scalar and
\begin{align}
\vec{P}_{\alpha}=
\left[
\begin{matrix}
\cos\alpha & -\sin \alpha & 0\\
\sin\alpha & \cos\alpha & 0\\
0&0&1
\end{matrix}
\right]
\label{eq:Palpha}
\end{align}
\begin{figure}[h]
\begin{center}
\includegraphics[width=84mm]{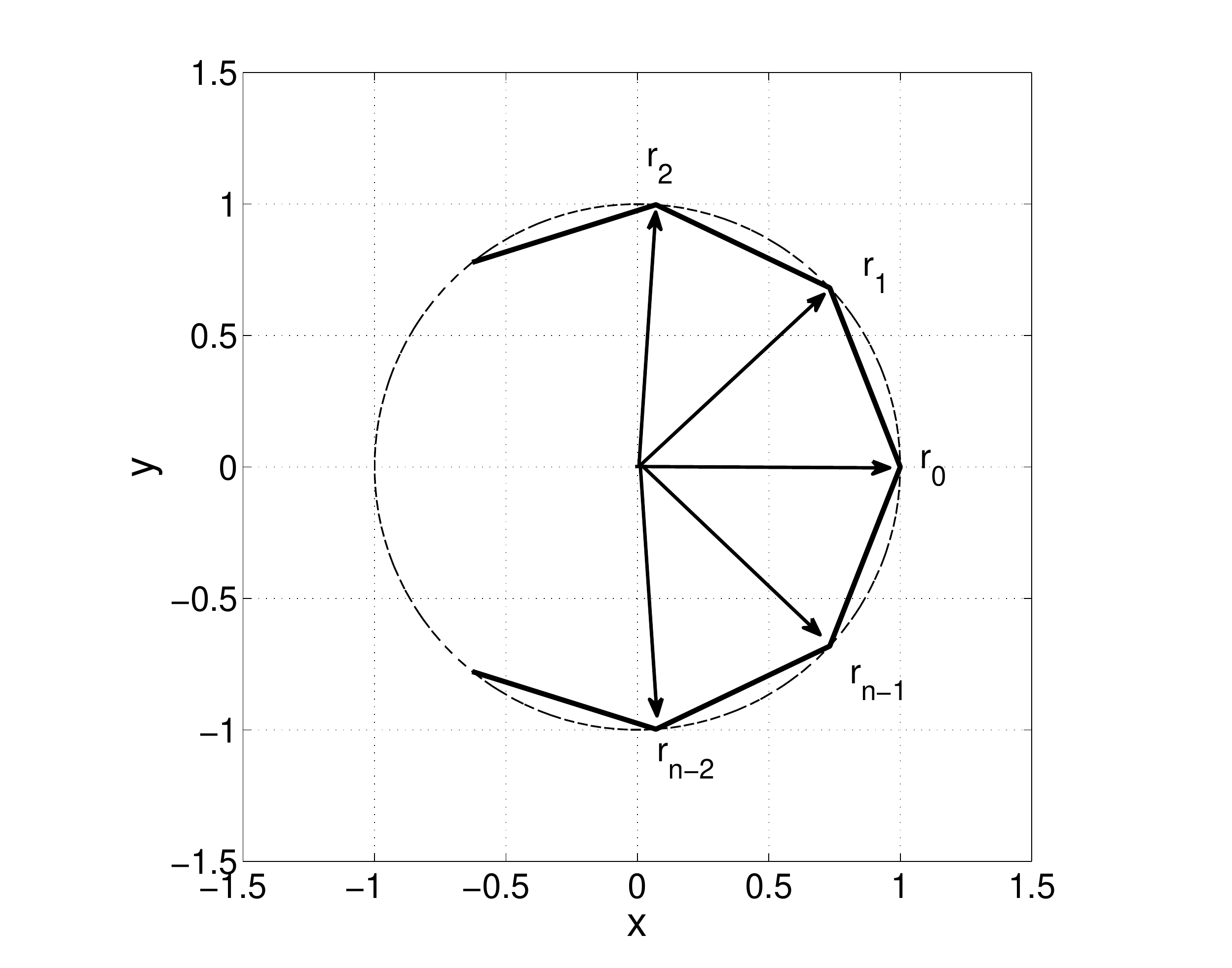}
\caption{Regular n-gon with vertices $\vec{r}_0...\vec{r}_{n-1}$}
\label{fig:sokszog}
\end{center}
\end{figure}
It is assumed that the principal radii of gyration reflect the symmetry of the set of contact points, i.e. $\rho_x=\rho_y$. 
We will exploit the symmetry of the object in order to simplify the analysis. Specifically, before an impact occurs (including the first one), we change labelling and reference frames. The exact transformations depend on the index $i$ of the vertex involved in the upcoming impact. Specifically, we perform the following two steps.

First, the local  coordinate system is rotated by angle $i\cdot 2\pi/n$ around the local $z$ axis and the global frame is rotated by the same angle about the global $Z$ axis. These transformations of the reference frames correspond to the following transformations of the state vectors and of the coordinates of vertices:
\begin{align}
\vec{q}&\rightarrow \vec{P}_{-2i\pi/n}\vec{q}\label{eq:Pq2}\\
\vec{p}&\rightarrow \vec{P}_{-2i\pi/n}\vec{p}\label{eq:Pp2}\\
\vec{r}_i^l&\rightarrow \vec{P}_{-2i\pi/n}\vec{r}_i^l \label{eq:Pri2}
\end{align}
All other parameters and equations governing the dynamics of the system remain unchanged.

Second, the labels of the vertices are shifted cyclically such that the previous vertex $i$ becomes vertex 0. This step corresponds to the inverse transformation of \eqref{eq:Pri2},
hence the two steps leave $\vec{r}_i^l$ unchanged. As a consequence of the technical steps, the next point to hit the ground will always be point 0. Thus, the combined effect of the technical steps and the impact to the generalized velocity is the linear transformation
\begin{align}
\vec{p}\rightarrow \vec{U}_0\vec{P}_{-2i\pi/n}\vec{p}
\label{eq:ptransformation}
\end{align} 
According to \eqref{eq:impactequ}, the impact map $\vec{U}_0$ can be expressed as
\begin{align}
\vec{U}_0=
\left[
\begin{matrix}
1 & 0 & 0\\
0 & 1 - \frac{\gamma + 1}{\rho^2 + 1} & \frac{\gamma + 1}{\rho^2 + 1}\\
0 & \frac{\rho^2(\gamma + 1)}{\rho^2 + 1} & \frac{\gamma + 1}{\rho^2 + 1}-\gamma
\end{matrix}
\right]
\label{eq:U0}
\end{align}

The generalized velocity of the rod after the first impact will be
\begin{align}
\vec{p}^{(1)}=\vec{U}_0 \vec{P}_{-2i_{0}\pi/n}\vec{p}^{(0)}, i_0\in\{0,1,...,n-1\}
\end{align}
whereas after $k$ impacts, we will have
\begin{align}
\begin{split}
\vec{p}^{(k)}&=(\vec{U}_0 \vec{P}_{-2i_{n-1}\pi/n})(\vec{U}_0 \vec{P}_{-2i_{n-2}\pi/n})...\\
&...(\vec{U}_0 \vec{P}_{-2i_1\pi/n})\vec{p}^{(1)} \label{eq:pk1}
\end{split}
\end{align}
where the integers $i_1,i_2,...,i_{n-1}\in\{1,2,...,n-1\}$ depend on the actual collision sequence. We can now formulate a sufficient condition of CC or PCC 

\begin{theorem}
If the contact points of $\mB$ form a regular $n$-gon and there exists a cone $\mathcal{K} \subset \mathbb{V}$ such that
\begin{enumerate}
\item   all $\vec{p}\in\mK$ satisfy \eqref{eq:vz<0}.
\item $\mK$ is a common invariant cone of the set of matrices 
$$
\{\vec{U}_0\vec{P}_{-2\pi/n},\vec{U}_0\vec{P}_{-4\pi/n},...,\vec{U}_0\vec{P}_{(-(2n-2)\pi/n}\}
$$
\item $\vec{p}^{(1)}\in\mK$
\end{enumerate} 
 then the object undergoes CC or PCC.
\label{thm:polygon1}
\end{theorem}

\begin{proof}
Conditions 2 and 3 imply that $\vec{p}^{(k)}\in\mK$ for all $k\geq 1$. At the same time, condition 1 and Lemma \ref{lemma:vz<0} imply the statement of the theorem.
\end{proof}

Unfortunately, it turns out that this result is very restrictive. Clearly, a set of matrices cannot have a common invariant cone unless all of them have have invariant cones individually. In our case, if $n$ is even, then the set of matrices includes $\vec{U}_0\vec{P}_{\pi}$. The dominant eigenvalue of $\vec{U}_0\vec{P}_{\pi}$ is $-1$, and thus Theorem \ref{thm:invconetheorem} implies that $\vec{U}_0\vec{P}_{\pi}$ does not posses an invariant cone. Hence, Theorem \ref{thm:polygon1} is in this case useless. We are in a similar situation in the case of odd $n$: there are large regimes in parameter space where at least one of the matrices has no invariant cone, nevertheless numerical simulations suggest that the object undergoes CC.  

\subsection{Constraints of collision sequences}\label{sec:order}

To improve the applicability of the common invariant cone approach, we now identify constraints of impact sequences during CC, and use the new concept of effectively invariant cones.

Recall that vertex $i$ touches the ground if and only if $\ \vec{f}_i^T \vec{q}=0$, where $\vec{f}_i$ is given by \eqref{eq:fi}. These points form a plane $\mathcal{F}_i$ in $\mathbb{C}$. The set of penetration-free configurations takes the form of a polyhedral cone $\mathcal{F}$, with $n$ facets.
$$
\vec{q}\in \mathcal{F} \iff \vec{f}_i^T \vec{q} \geq 0 \quad \vert \ \forall i \in \{0,1...n-1\}
$$
Fig. \ref{fig:conditions}(a,c,e) illustrates this cone for $n=4$. If an impact at vertex $0$ is followed by an impact at $j$, then the system moves from an initial configuration $\vec{q}_0 \in \mathcal{F}_0 \cap \mathcal{F}$ to a final configuration $\vec{q}_j \in \mathcal{F}_j\cap \mathcal{F}$ along a straight trajectory, i.e. 
\begin{align}
\vec{q}_j=\vec{q}_0+\tau \vec{p} \quad \vert \ \tau>0, \ \vec{p} \in \mathbb{V}
\label{eq:i>>j}
\end{align}
Here $\vec{p}$ is the (approximately constant) generalized velocity of the body between the two collisions and $\tau$ is the time spent between the collisions. According to Lemma \ref{lemma:vz<0}, the trajectory also satisfies \eqref{eq:vz<0}.

Those values of $\vec{p}$ for which such a trajectory exists, form a cone $\mathcal{C}_{j}$

We now define the  point $\vec{m}_{k,l}$ ($k,l\in{0,1,...,n-1}$) in configuration space, as the solution of the three equations:
\begin{align}
\vec{f}_k^T \vec{m}_{k,l}=0 \label{mij1} \\  
\vec{f}_l^T \vec{m}_{k,l}=0 \label{mij2} \\
\vec{u}_3^T \vec{m}_{k,l}=1\label{mij3}
\end{align}
The point $\vec{m}_{k,l}$ corresponds to a configuration in which the distance of the centre of mass from $\mP$ is 1, while vertex $k$ and vertex $l$ are in contact with the ground
Then,
\begin{lemma}\label{lem:cone}
The cone $\mathcal{C}_{j}$ is generated by four vectors 
$$\{(\vec{m}_{j-1,j}-\vec{m}_{n-1,0}),-\vec{m}_{n-1,0},-\vec{m}_{0,1},(\vec{m}_{j,j+1}-\vec{m}_{0,1})\}$$ 
\end{lemma}  
The proof of Lemma \ref{lem:cone} is given in the Appendix. 
One can use cross products of adjacent generating vectors to construct the inward pointing normals of the four facets of  $\mathcal{C}_{j}$. Thus, $\vec{p}\in \mathcal{C}_{j}$ if and only if all of the following relations hold:
\begin{align}
(-\vec{m}_{n-1,0} \times -\vec{m}_{0,1})^T \vec{p}= \vec{f}_0^T \vec{p}  &\geq 0 \label{eq:felt1} \\
\begin{split}
((\vec{m}_{j,j+1}-\vec{m}_{0,1})\times (\vec{m}_{j-1,j}-\vec{m}_{n-1,0}))^T \vec{p} &=\\
-\vec{u}_3^T \vec{p}&\geq 0
\end{split} \label{eq:felt2}\\
(\vec{m}_{j,j+1} \times \vec{m}_{0,1})^T \vec{p} &\geq 0 \label{eq:felt3} \\ 
(\vec{m}_{n-1,0} \times \vec{m}_{j-1,j})^T \vec{p}  &\geq 0 \label{eq:felt4}
\end{align}

Condition (\ref{eq:felt1}) means that vertex 0 goes upwards after collision. Condition (\ref{eq:felt2}) is indeed equivalent of \eqref{eq:vz<0}. The remaining two inequalities are nontrivial necessary conditions for the $j$-th vertex to collide before any other vertex does. These four conditions will be crucial for our main results.

As an example consider the case $n=4$ of a square-shaped object. In this case, $\vec{r}_0^l=[1\quad 0\quad 0]^T$; $\vec{r}_1^l=[0\quad 1\quad 0]^T$;$\vec{r}_2^l=[-1\quad 0\quad 0]^T$ and $\vec{r}_3^l=[0\quad -1\quad 0]^T$. Then, $\vec{f}_i$ is given by \eqref{eq:fi}, and \eqref{mij1}-\eqref{mij3} yield
\begin{align}
\vec{m}_{0,1}&=[-1\;1\;1]^T\label{eq:squarefelt1}\\
\vec{m}_{1,2}&=[-1\;-1\;1]^T\\
\vec{m}_{2,3}&=[1\;-1\;1]^T\\
\vec{m}_{3,0}&=[1\;1\;1]^T\label{eq:squarefelt4}
\end{align}

Figure \ref{fig:conditions}(a,c,e) illustrate the planes $\mathcal{F}_i$, the points $\vec{m}_{i,j}$, and the generating vectors of the cones $\mathcal{C}_i$ for $i=1,2,3$ in configuration space. Panels (b,d,f) of the figure depict the cones $\mathcal{C}_i$ in velocity space.

 \begin{figure*}[t!]
    \centering
    \begin{subfigure}[t]{0.5\textwidth}
        \centering
        \includegraphics[width=60mm]{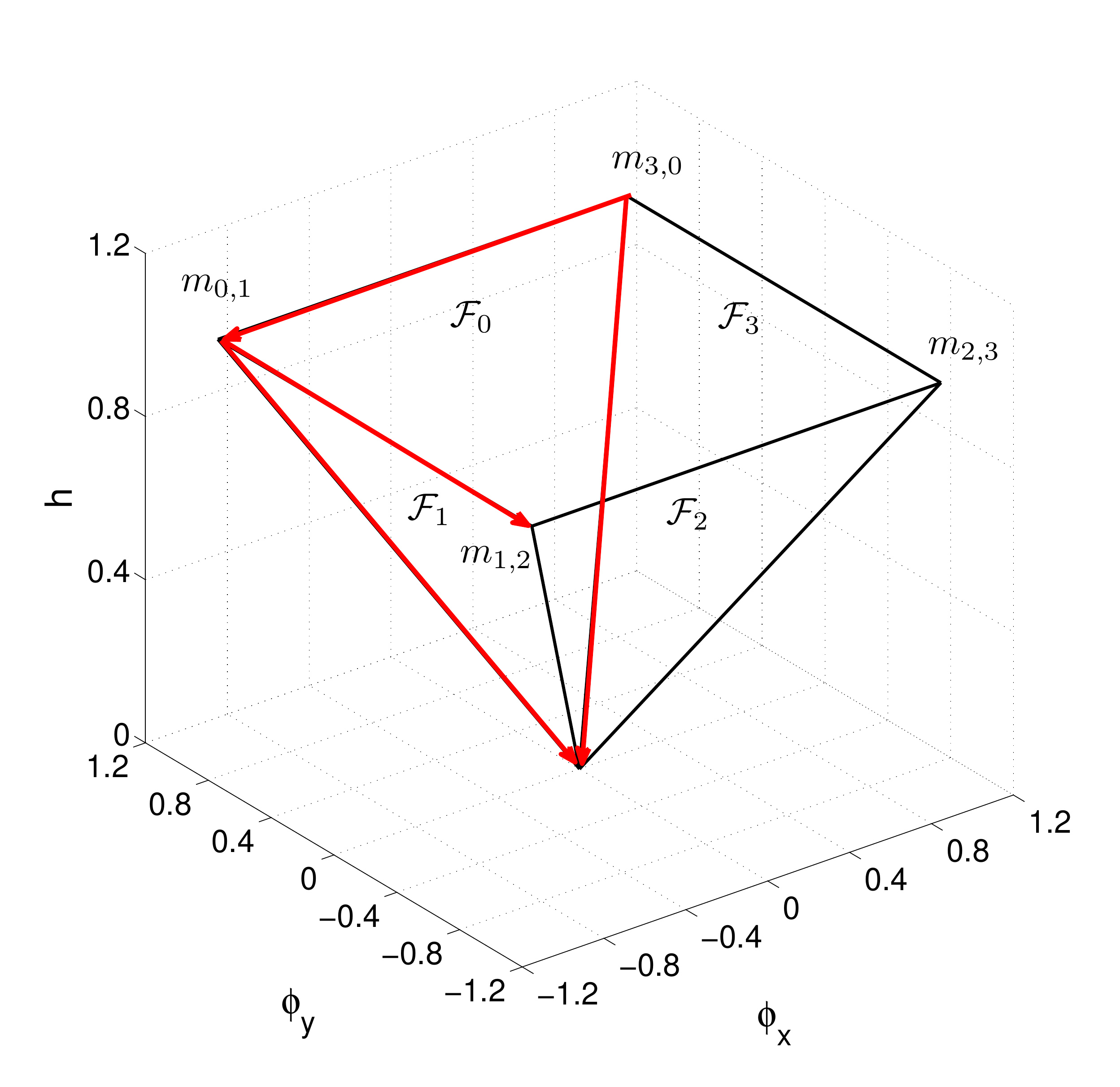}
        \caption{}
    \end{subfigure}%
    ~ 
    \begin{subfigure}[t]{0.5\textwidth}
        \centering
        \includegraphics[width=60mm]{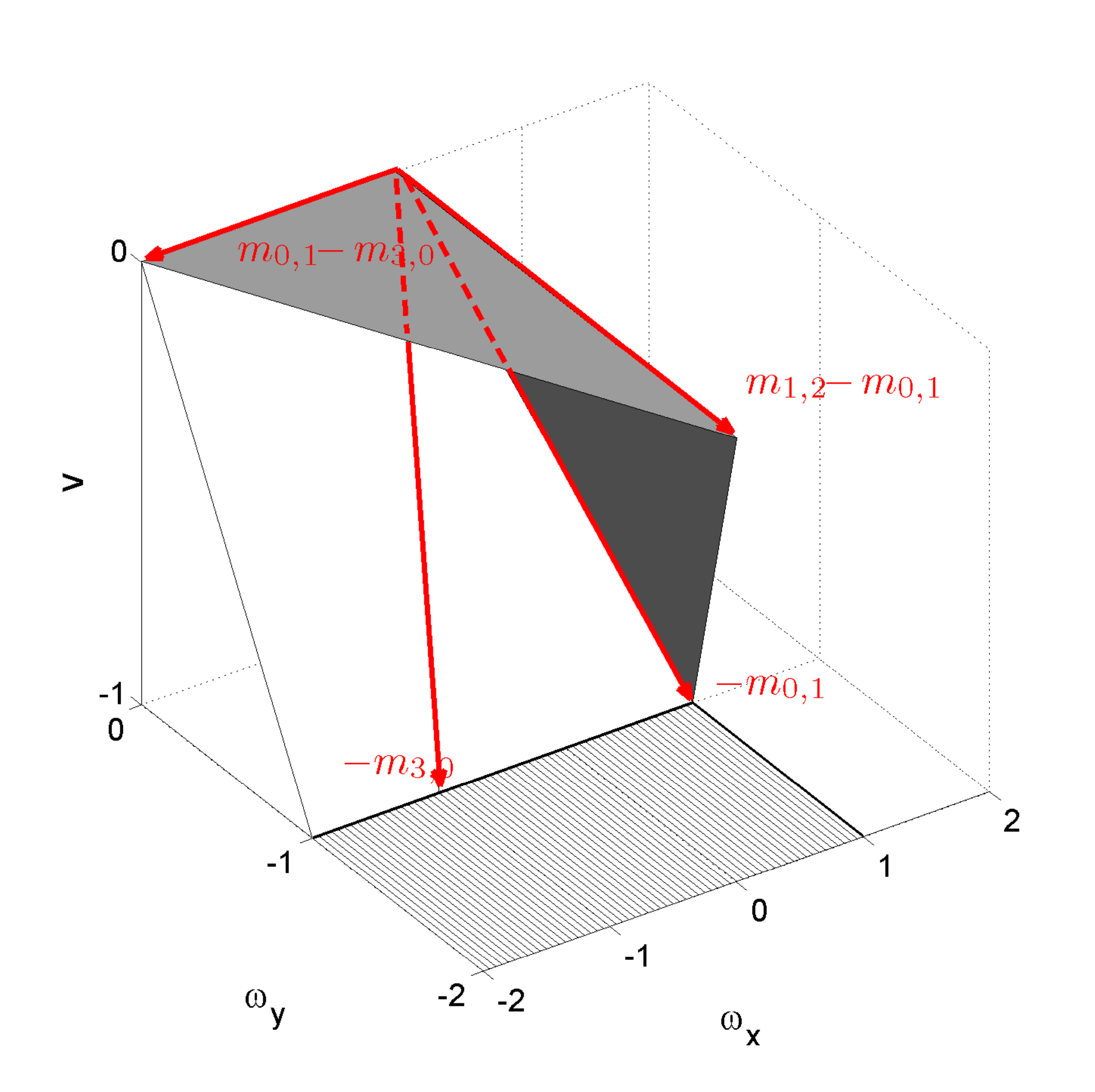}
        \caption{}
    \end{subfigure}
    \\
    \begin{subfigure}[t]{0.5\textwidth}
        \centering
        \includegraphics[width=60mm]{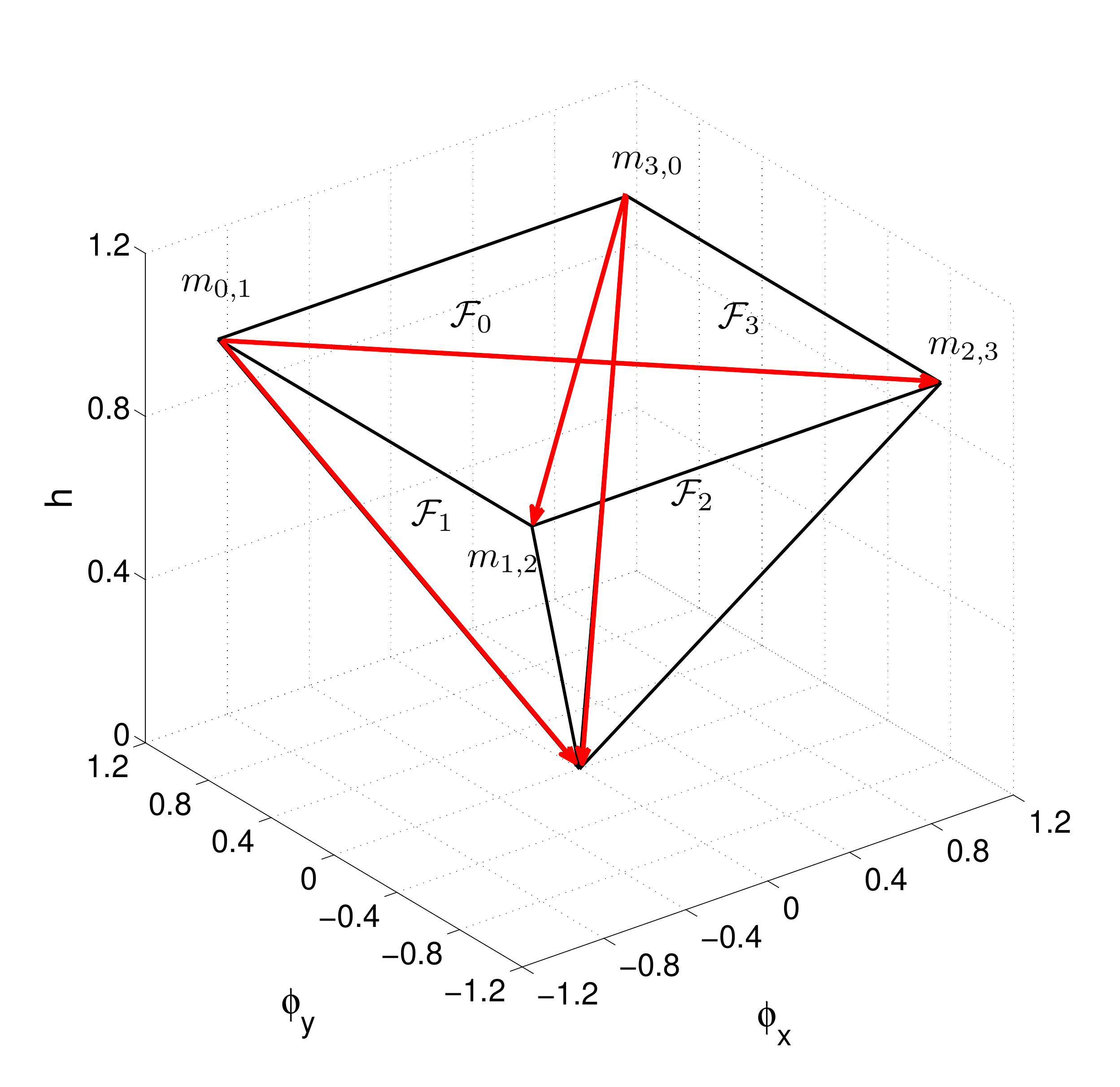}
        \caption{}
    \end{subfigure}%
    ~ 
    \begin{subfigure}[t]{0.5\textwidth}
        \centering
        \includegraphics[width=60mm]{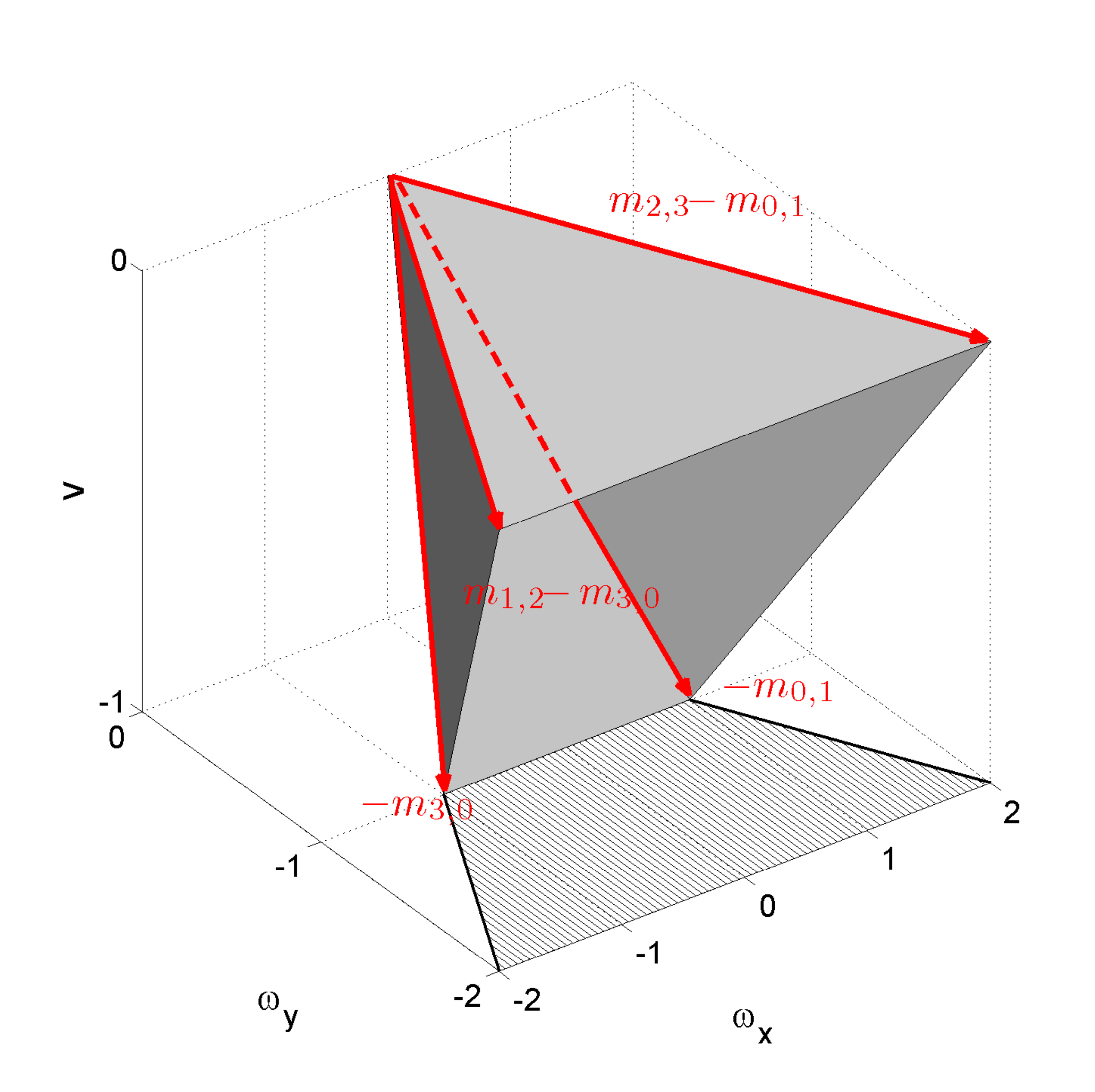}
        \caption{}
    \end{subfigure}
    \\
    \begin{subfigure}[t]{0.5\textwidth}
        \centering
        \includegraphics[width=60mm]{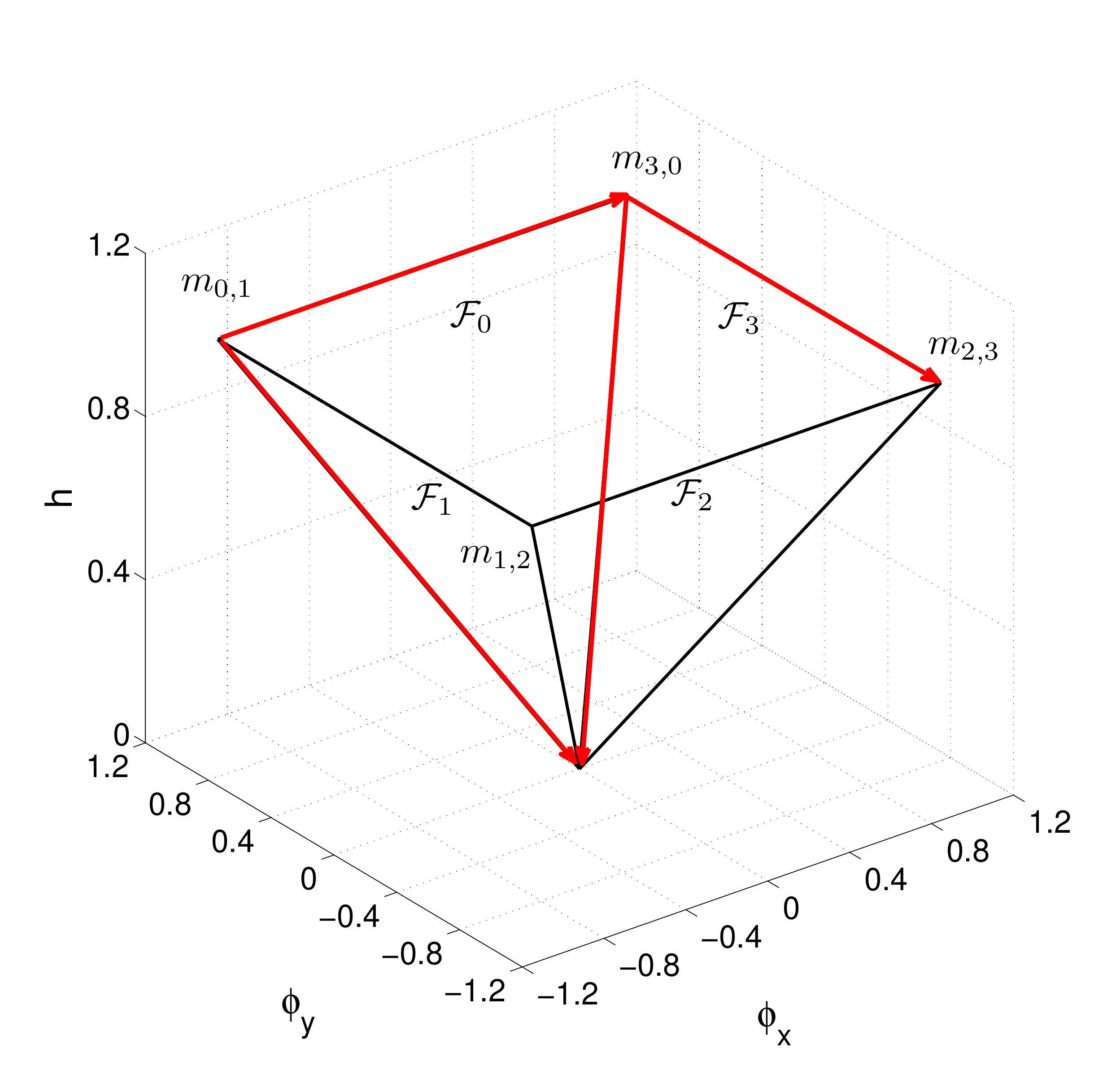}
        \caption{}
    \end{subfigure}%
    ~ 
    \begin{subfigure}[t]{0.5\textwidth}
        \centering
        \includegraphics[width=60mm]{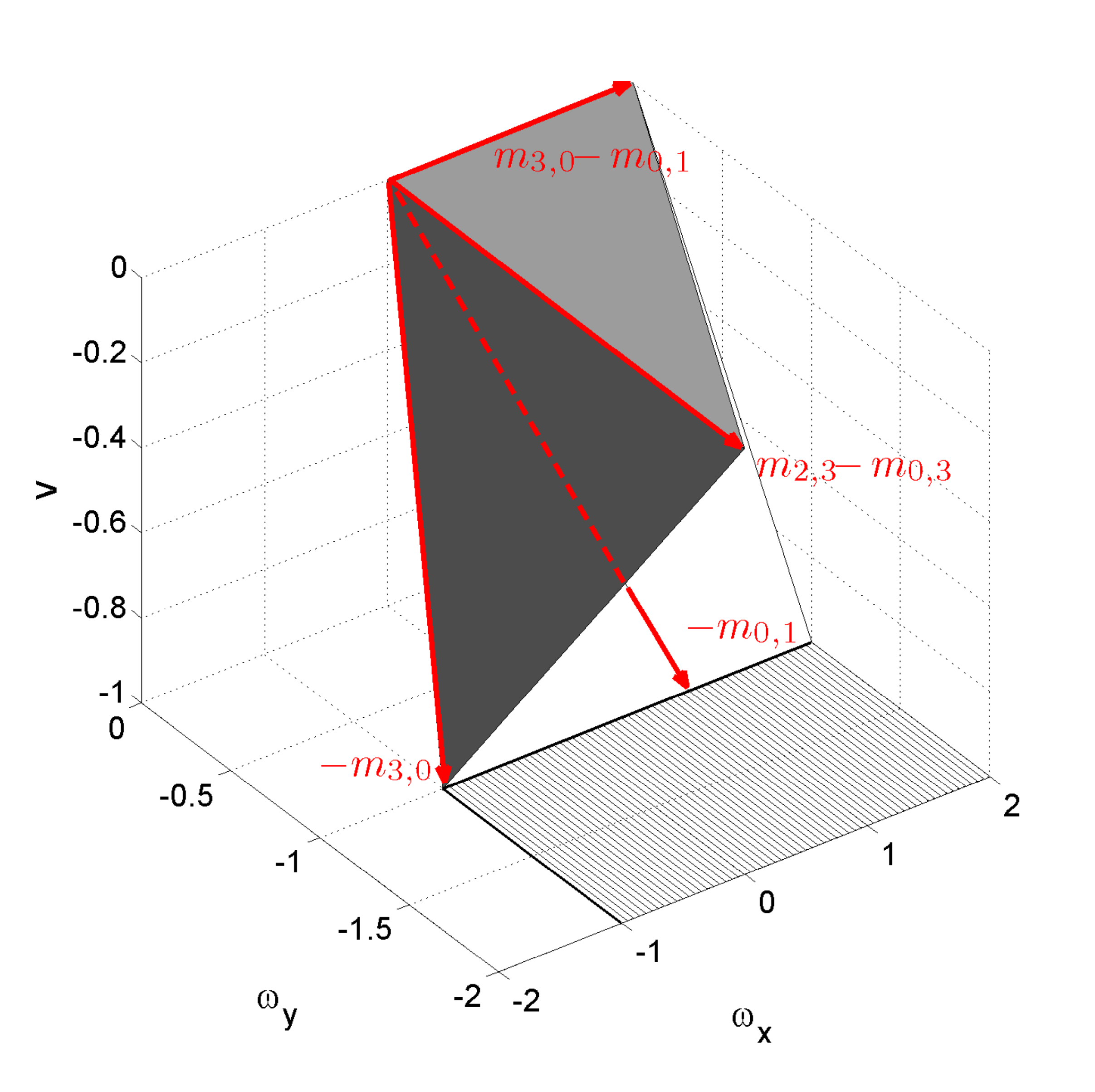}
        \caption{}
    \end{subfigure}
    \caption{Left: the penetration-free cone in configuration space with the generating vectors of $\mathcal{C}_1$ (a), $\mathcal{C}_2$ (c) and $\mathcal{C}_3$ (e). Right: the same three cones in velocity space and their intersections with the $v=-1$ plane (hatching), which is used in Sec. \ref{sec:effinvcone}}
    \label{fig:conditions}
\end{figure*}

\subsection{A stronger sufficient condition}
We have already developed a sufficient condition of CC (Theorem \ref{thm:polygon1}) in Sec. \ref{sec:commoncone}, which is applicable when the contact points form a regular $n$-gon. Nevertheless, we have seen that the sufficient condition is too restrictive and thus useless. To overcome this difficulty, constraints on impact sequences have been developed in Sec. \ref{sec:order}. We can combine these two results into a stronger sufficient condition of CC:

\begin{theorem} \label{thm:polygon2}
If the contact points of $\mB$ form a regular $n$-gon and there exists a cone $\mathcal{K} \in \mathbb{V}$ such that
\begin{enumerate}
\item   all $\vec{p}\in\mK$ satisfy \eqref{eq:vz<0} 
\item $\mK$ is an effectively invariant cone of the set of matrices 
$$
\{\vec{U}_0\vec{P}_{-2\pi/n},\vec{U}_0\vec{P}_{-4\pi/n},...,\vec{U}_0\vec{P}_{(-(2n-2)\pi/n}\}
$$
and the conditions $\{\mathcal{C}_1$, $\mathcal{C}_2$,...,$\mathcal{C}_{n-1}\}$ defined above
\item $\vec{p}^{(1)}\in\mK$
\end{enumerate} 
 then the object undergoes CC or PCC.
\end{theorem}
The proof is identical to that of Theorem \ref{thm:polygon1}. 

\subsection{Partial complete chatter}
The most important limitation of Lemma \ref{lemma:vz<0}, Theorem \ref{thm:polygon1}, and Theorem \ref{thm:polygon2} is that they  cannot distinguish between CC and PCC. We will now fill this gap by showing that PCC can be outruled in most cases, which turnes Theorem \ref{thm:polygon2} into a sufficient condition of CC. 

\begin{lemma}
If the conditions of Theorem \ref{thm:polygon2} are satisfied then 
PCC is impossible unless the matrix $\vec{U}_0\vec{P}_{2\pi/n}\vec{U}_0\vec{P}_{-2\pi/n}$ has  real eigenvalues.
\label{lem:noPCC}
\end{lemma}
\begin{proof}
There are two possible ways for $\mB$ to undergo PCC: 
\begin{enumerate}
\item A pair of non-adjacent vertices collide with $\mathcal{P}$ in alternating order. 
\item A pair of adjacent vertices hit $\mathcal{P}$ in alternating order.
\end{enumerate}

Assume that the first scenario occurs. At the end of the PCC sequence, the two non-adjacent vertices involved in the impact sequence rest in contact with $\mathcal{P}$ with 0 velocity. The only point within the cone $\mathcal{K}$ with this property is its tip, i.e. $\vec{p}=[0,0,0]^T$. Hence $\mathcal{B}$ is indeed immobile, which means that the object has undergone CC instead of PCC.

Consider now the second scenario. Assume that the object enters a PCC sequence after an initial transient with $t$ impacts and the two vertices involved in the PCC sequence are labelled at this point as 0 and 1. Then the generalized velocity after $t+2k$ impacts will be  
\begin{align}
\vec{p}^{(t+2k)}=(\vec{U}_0\vec{P}_{2\pi/n}\vec{U}_0\vec{P}_{-2\pi/n})^{k}\vec{p}^{(t)}
\label{eq:PCCmap}
\end{align}
Note that \eqref{eq:PCCmap} is highly analogous to \eqref{eq:pk1rod} in the rod problem. Indeed the PCC of a polygon is very similar to the CC of a rod.

The dominant eigenvalue of matrix $\vec{U}_0\vec{P}_{-2\pi/n}\vec{U}_0\vec{P}_{2\pi/n}$  is $+1$ and the corresponding eigenvector corresponds to a generalized velocity for which bothvertices involved in the PCC sequence are immobile, i.e. $v_0=v_1=0$ . For large values of $k$, \eqref{eq:PCCmap} implies that the direction of $\vec{p}^{(t+2k)}$ converges to this eigenvector. At the same time, if the other two eigenvalues are complex, then we are in a situation similar to case 3 in the proof of Theorem \ref{thm:goyal}: $\vec{p}^{(t+2k)}$ spirals around the dominant eigenvector, and there will be values of $k$ for which vertex 1 moves upwards and cannot hit $\mathcal{P}$ hence the PCC sequence cannot continue. This contradiction indicates that PCC is not possible unless all eigenvalues are real.     
\end{proof}

For simplicity we omit the detailed investigation of the matrix $\vec{U}_0\vec{P}_{-2\pi/n}\vec{U}_0\vec{P}_{2\pi/n}$. It turns out that the conditions of Lemma \ref{lem:noPCC}  are never satisfied if $n\geq 4$ and they are not satisfied but a small region of the space of physical parameters in the case of $n=3$: this region is bounded by a dashed curve     and is labelled as "PCC possible" in Fig. \ref{fig:simulations}(a). 

\section{Constructing effectively invariant cones}\label{sec:effinvcone}

In this section, we will construct effectively invariant cones which satisfy the conditions of Theorem \ref{thm:polygon2}, thereby we will develop sufficient conditions of CC. The sharpness of these results will be tested by systematic numerical simulations.

\subsection{A numerical construction for arbitrary $n$}
\label{sec:num}

Numerical approximations of effectively invariant cones can be constructed by iterative algorithms. The algorithm outlined below considers an initial candidate, which is gradually increased by taking the union of the candidate cone with its transformed images, until the sequence of candidate cones converges to an effectively invariant cone, or until \eqref{eq:vz<0} is violated. The detailed steps are as follws:

\begin{enumerate}
\item We choose an initial set of vectors in $\mathbb{V}$, all of which satisfy \eqref{eq:vz<0} and \eqref{eq:felt1}. These vectors generate an initial candidate cone $\mK_0$. The requirement \eqref{eq:felt1} is inspired by the fact that any post-impact velocity must satisfy \eqref{eq:felt1}. 

\item Given a polyhedral candidate cone $\mK_k$, one can construct the generating vectors of the cones $\mK_{k,i}:=\mK_k\cap\mathcal{C}_i$ for $i=1,2,...,n-1$. This step is straightforward since the cones $\mathcal{C}_i$ are also polyhedral.  
\item The generating vectors of the transformed cones $\mK_{k+0.5,i}:=\vec{U}_0\vec{P}_{-2i\pi/n}(\mK_{k,i})$ are constructed by transformation of each individual generating vector of $\mK_{k,i}$. 
\item If any of these vectors violate \eqref{eq:vz<0}, then the algorithm terminates with the conclusion that an effectively invariant cone satisfying \eqref{eq:vz<0}, with $\mK_0$ in its interior does not exist
\item A next candidate cone is constructed:
$$
\mK_{k+1}=\left(\cup_{i=1}^{n-1}\mK_{k+0.5,i}\right)\cup\mK_k
$$ 
Technically speaking, all generating vectors involved in the union operations are normalized by the transformation $\vec{p}\rightarrow \vec{p}/(\vec{u}_3^T\vec{p})$ and the generating vectors of $\mK_{k+1}$ are obtained by finding the convex hull of the normalized set of vectors.
\item If $\mK_{k+1}\subseteq\mK_{k}$ is satisfied or if the solid angle of the cone 
$\mK_{k+1}\backslash\mK_{k}$ is below a tolerance parameter $\epsilon$, then the algorithm terminates with success and $\mK_{k+1}$ is deemed to be an $\epsilon-$approximation of $\mK$.

\item The algorithm continues with step 2.
\end{enumerate}
The algorithm always terminates in finite number of steps, since the solid angle of the candidate cone always increases by at least $\epsilon$, and it cannot exceed $2\pi$ (i.e. the solid angle associated with a half-space). 
False negatives are avoided, i.e. if there exists an effectively invariant cone and the initial candidate is in its interior, then the algorithm always terminates with a positive answer. False positive results are however likely to occur for relatively large values of  $\epsilon$. We believe that false positive results must disappear for any set of transformations, if $\epsilon$ is sufficiently small, but a formal proof of this statement is beyond the scope of the paper. In what follows, we use $\epsilon=10^{-5}$.

If the main cycle of the algorithm is repeated many times, the number of generating vectors of $\mK_k$ and thus the computational cost of every step may increase rapidly. (In particular, the candidate cones often converge to a cone bounded by a smooth curve). Thus, in practice, we terminate the algorithm after a limited number of iterations (typically around $10^2$), without a conclusive answer. 

\begin{figure*}[t!]
    \centering
    \begin{subfigure}[t]{0.5\textwidth}
    \centering
    \includegraphics[width=60mm]{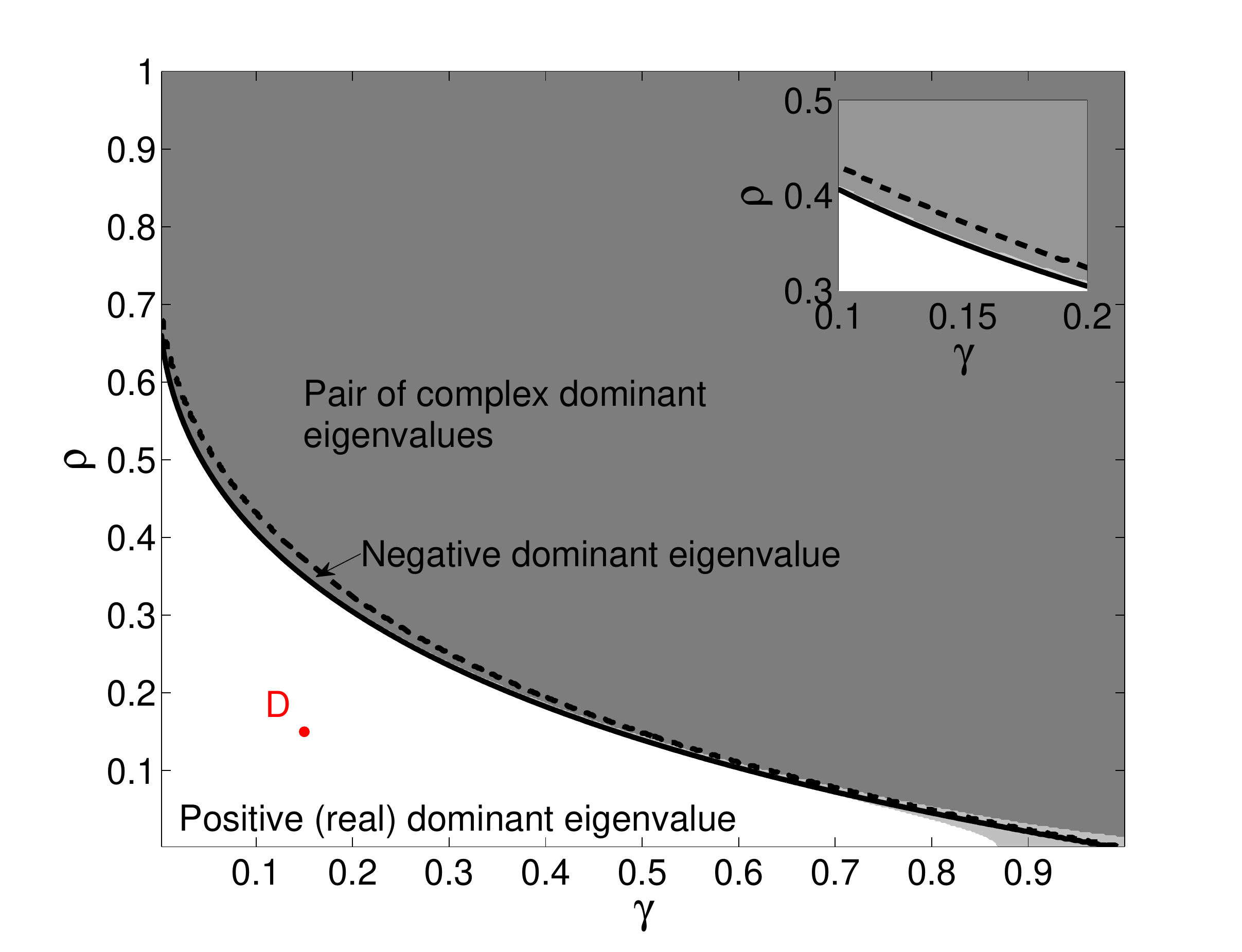}
    \caption{}
    \end{subfigure}%
    ~ 
    \begin{subfigure}[t]{0.5\textwidth}
    \centering
    \includegraphics[width=60mm]{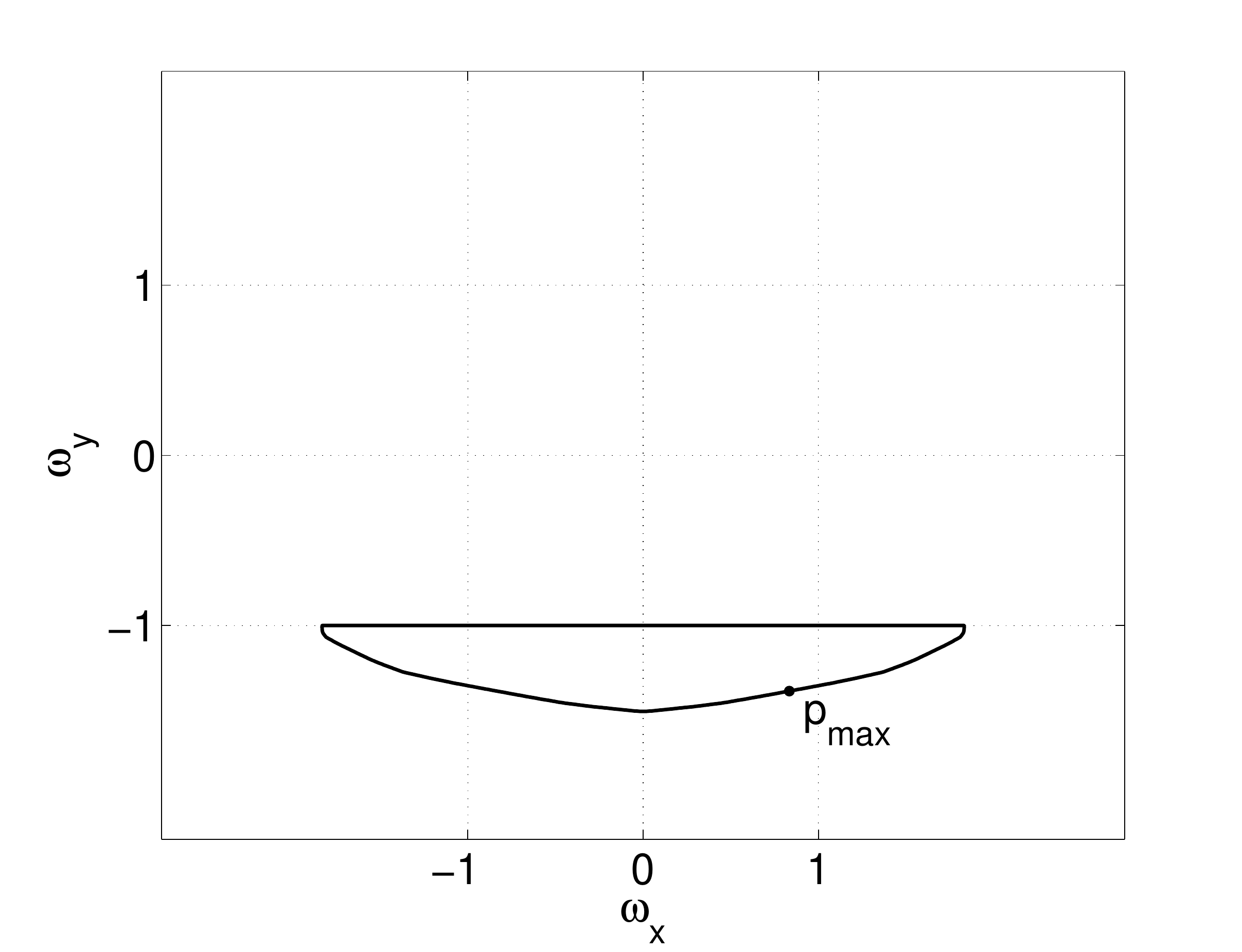}
    \caption{}
    \end{subfigure}%
    \\
    \begin{subfigure}[t]{0.5\textwidth}
    \centering
    \includegraphics[width=60mm]{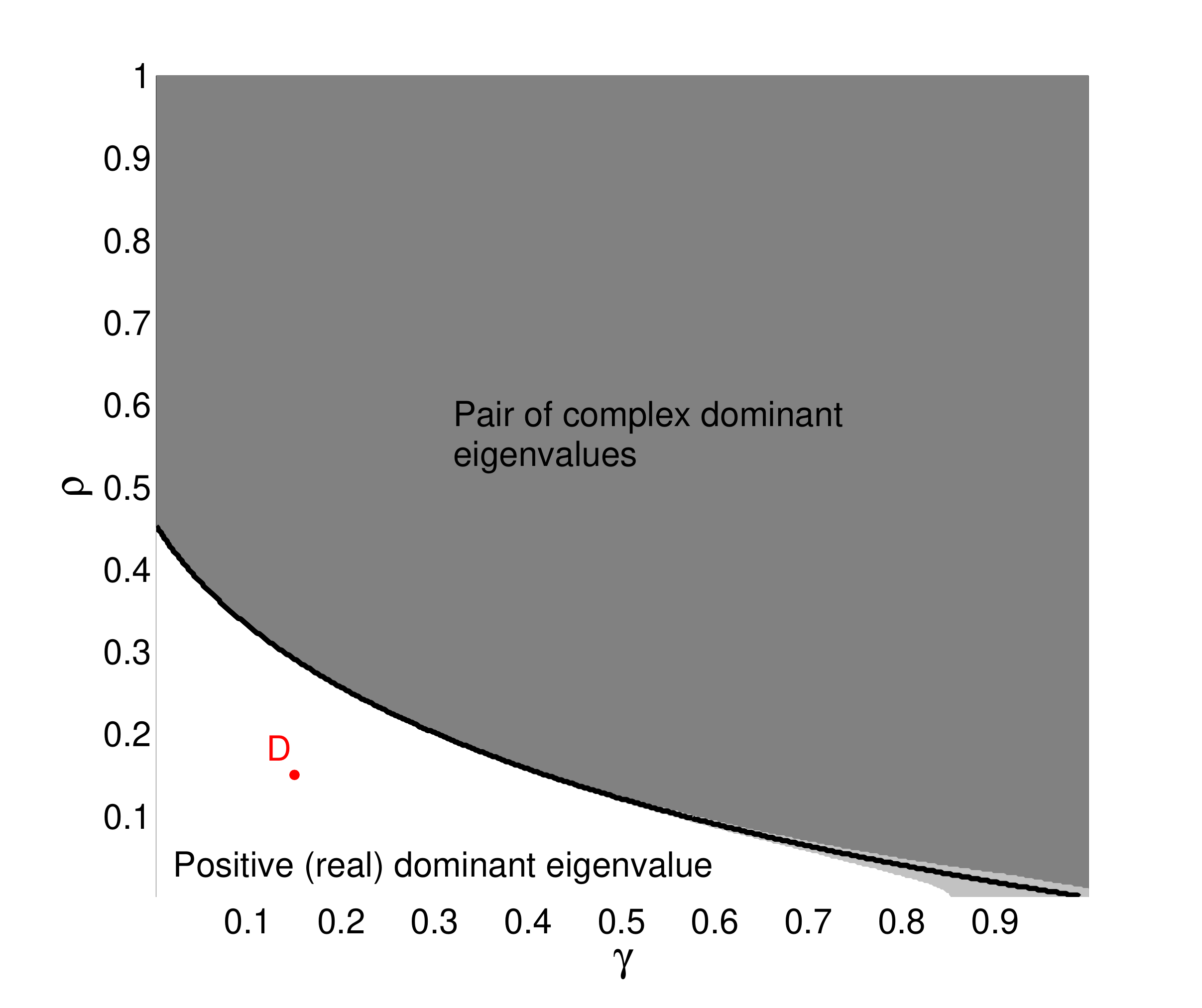}
    \caption{}
    \end{subfigure}%
    ~
    \begin{subfigure}[t]{0.5\textwidth}
    \centering
          \includegraphics[width=60mm]{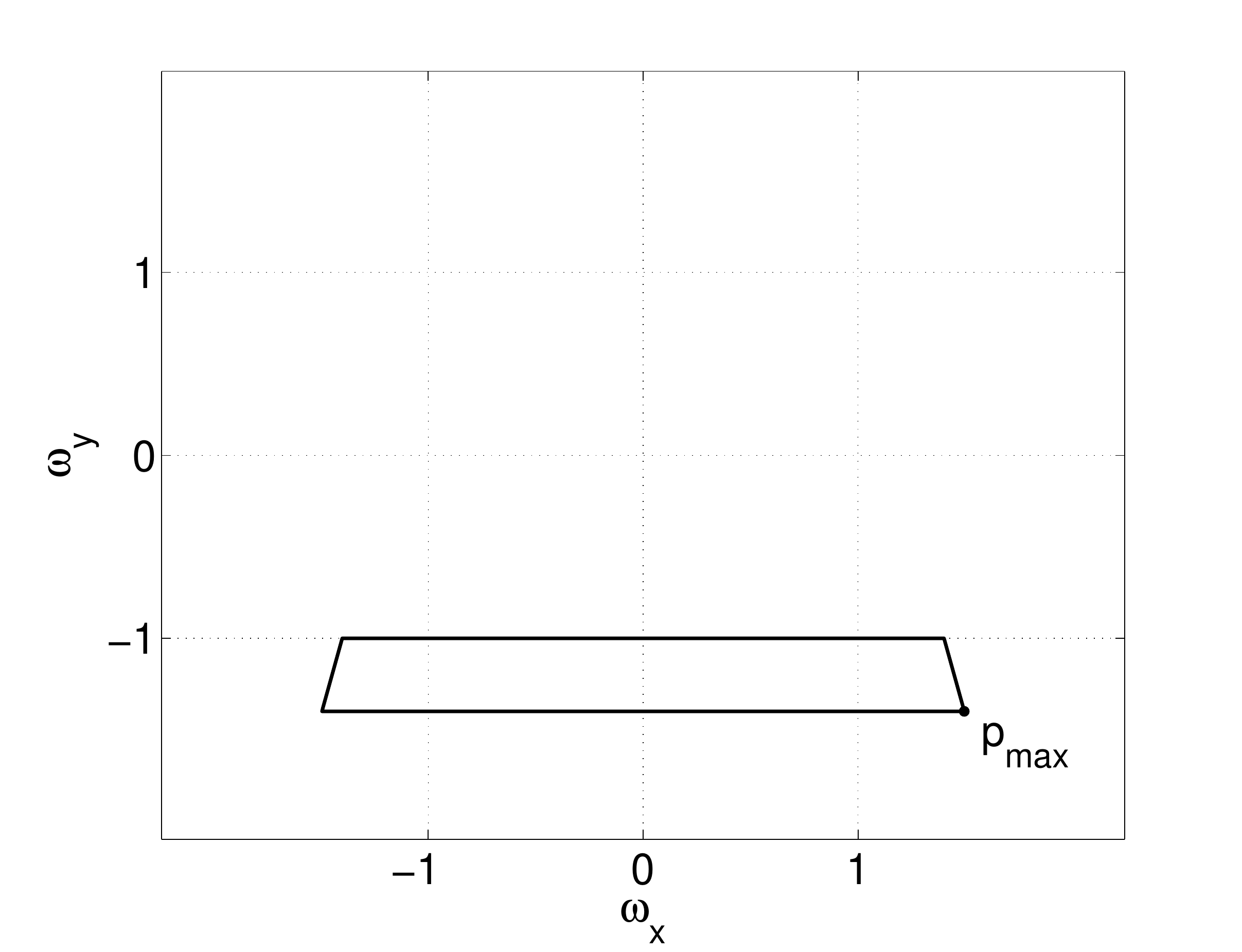}
          \caption{}
   \end{subfigure}%
    \\
    \begin{subfigure}[t]{0.5\textwidth}
    \centering
    \includegraphics[width=60mm]{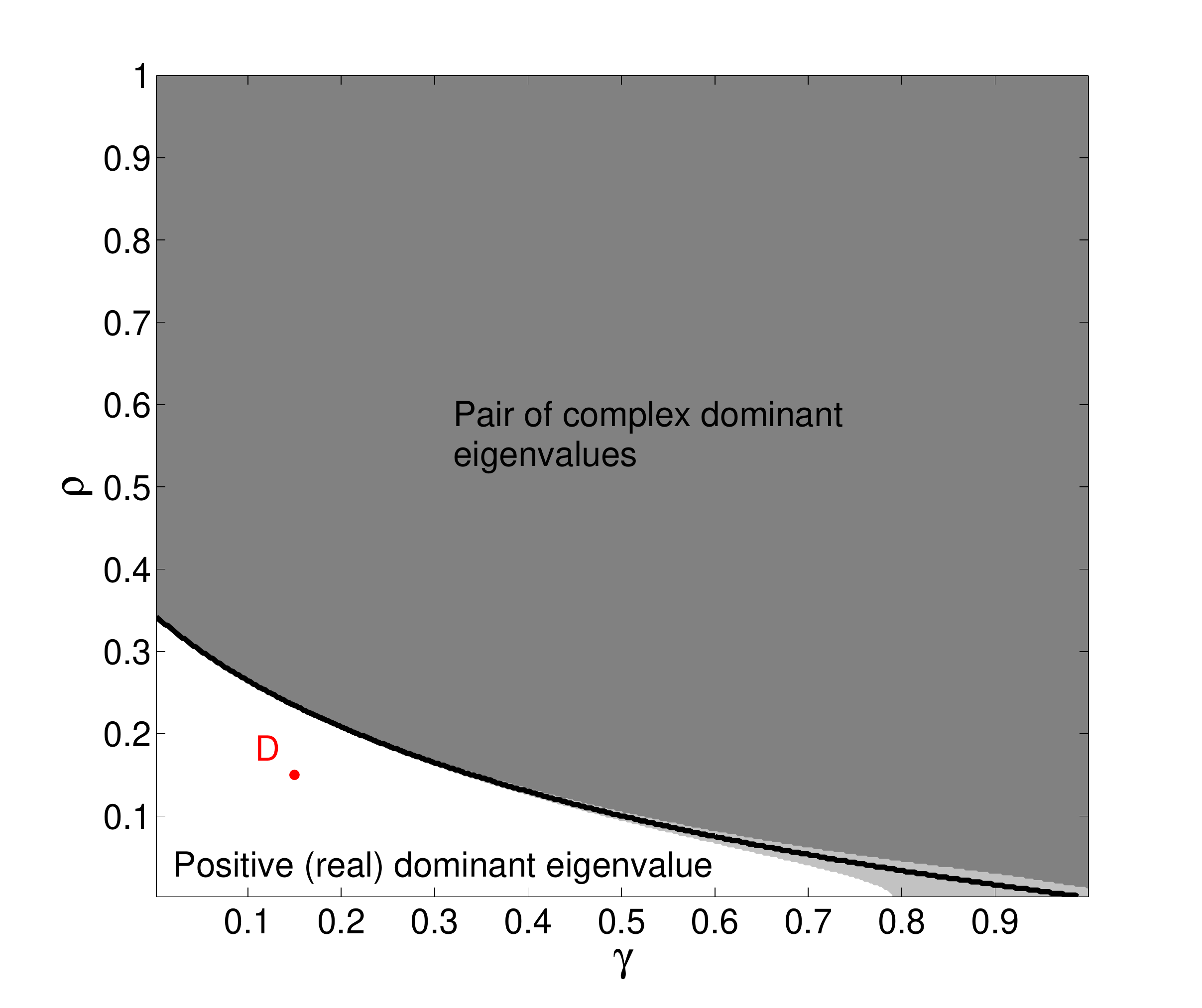}
    \caption{}
    \end{subfigure}%
    ~
    \begin{subfigure}[t]{0.5\textwidth}
    \centering
    \includegraphics[width=60mm]{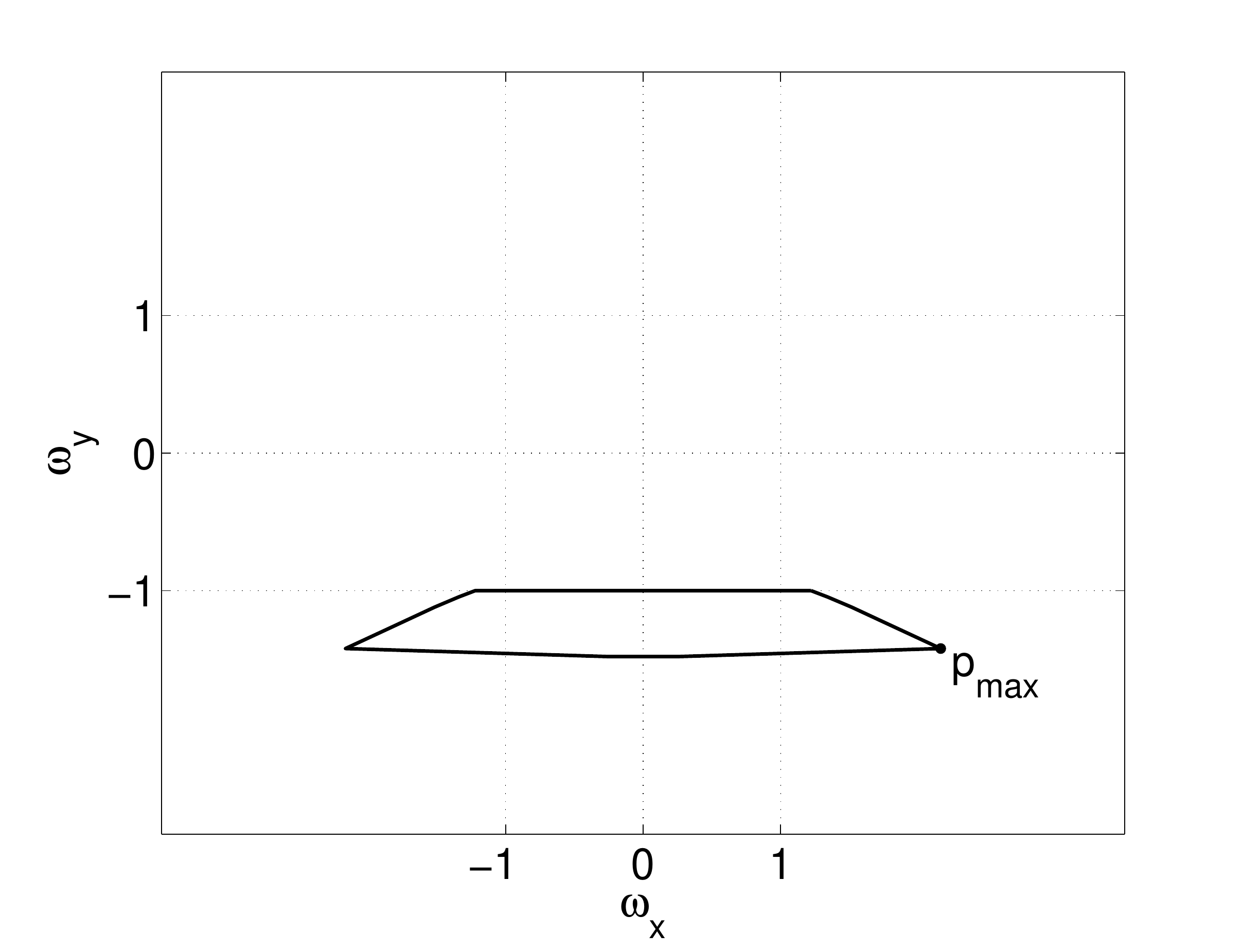}
    \caption{}
    \end{subfigure}%
    
    \caption{Effectively invariant cones for regular polygons. Left: The background color shows the results of the numerical cone-finding algorithm (dark: no effectively invariant cone exists, light: inconclusive, white:  effectively invariant cone has been found). Solid curves indicate those points of the parameter plane where the dominant eigenvalue(s) of $\vec{U}_0\vec{P}_{2\pi/3}$ (a), $\vec{U}_0\vec{P}_{2\pi/4}$ (c) and $\vec{U}_0\vec{P}_{2\pi/5}$ (e) change sign of become complex. Right: numerically found effectively invariant cones for a triangle (c), a square (d), and a pentagon (f). The parameter values corresponding to these cones are depicted by point D on the left side. $\vec{p}_{max}$  denotes the eigenvector corresponding to the dominant eigenvalue. The role of this eigenvector will be discussed in Sec. \ref{sec:square} in detail.}
    \label{fig:sajatertekek}
\end{figure*}

The cone-finding algorithm outlined above has been applied to $n$-gons with $n=3,4$ and $5$. In every case, the (degenerate) cone generated by the single vector $-\vec{U}_0\vec{u}_3$ was used as initial candidate. Three examples of effectively invariant cones recovered by the algorithm are shown in the right panels of Fig.\ref{fig:sajatertekek}, which shows a central projection of $\mathbb{V}$ to a plane $\mathcal{S}$ determined by the relation $\vec{u}_3^T\vec{p}=-1$. The projected image of $\mK$ is a polygon. We then run the algorithm for many values of the parameters $\rho$ and $\gamma$ along a rectangular grid. The results are summarized in the left panels of the figure: white means success, dark grey means that the invariant cone does not exist, and the small light grey regions near the bottom-right corners mean that the algorithm terminated before reaching a conclusion. The solid and dashed curves of the figure will be defined later.

\begin{figure}
    \centering
    \begin{subfigure}{0.5\textwidth}
        \centering
        \includegraphics[width=70mm]{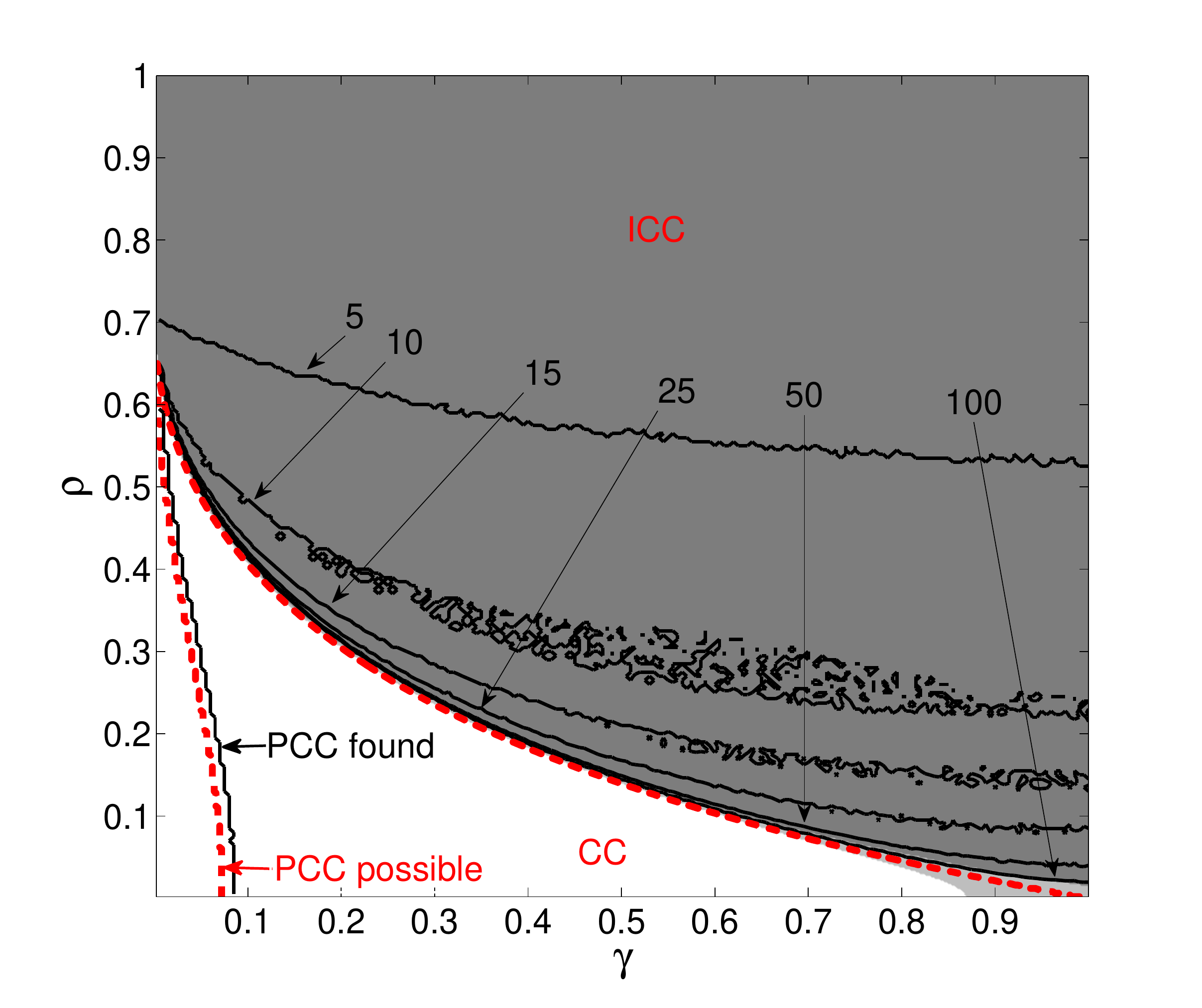}
        \caption{}
    \end{subfigure}%
    \\
        \centering
    \begin{subfigure}{0.5\textwidth}
        \centering
        \includegraphics[width=70mm]{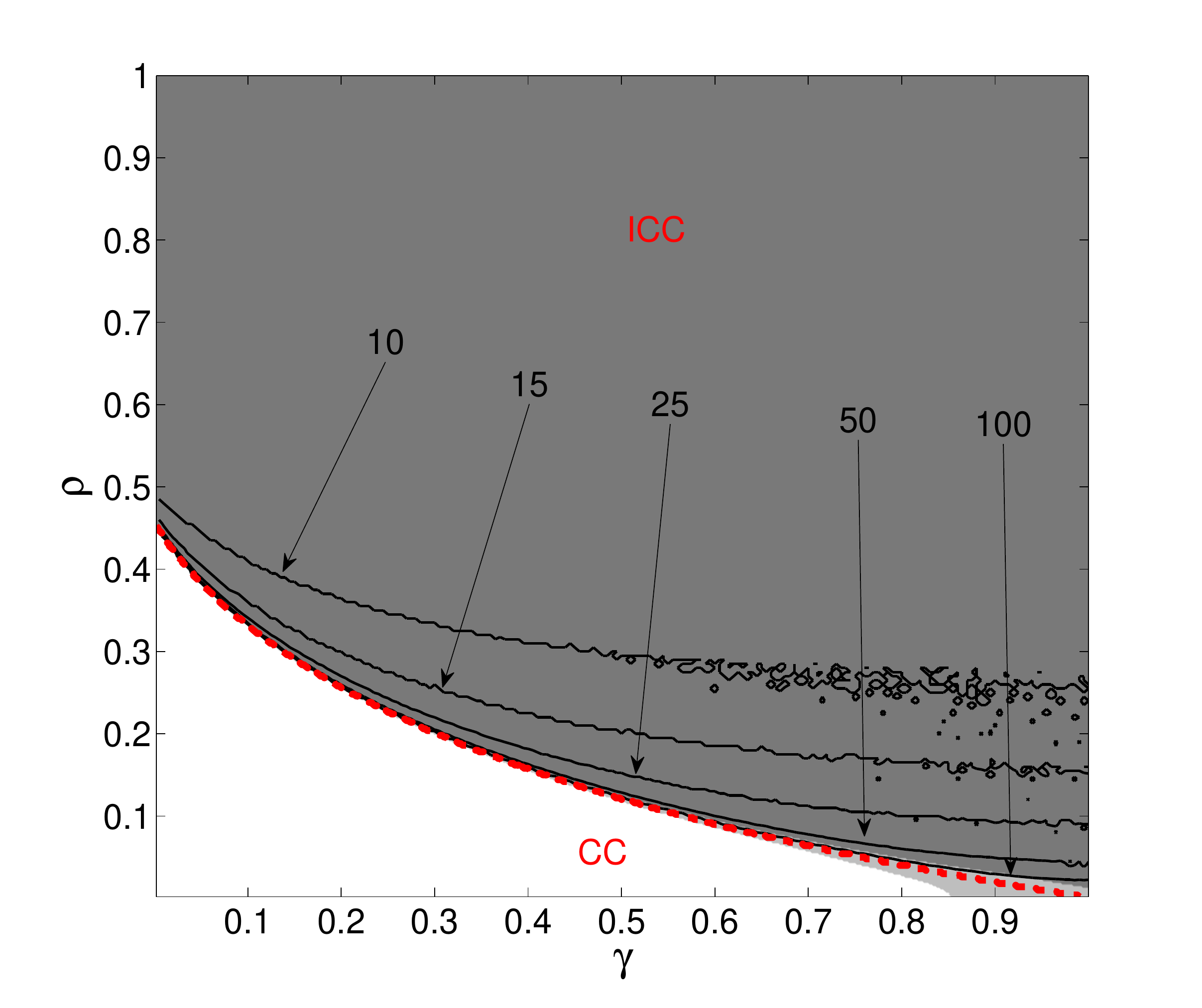}
        \caption{}
    \end{subfigure}%
    \\
    \begin{subfigure}{0.5\textwidth}
        \centering
        \includegraphics[width=70mm]{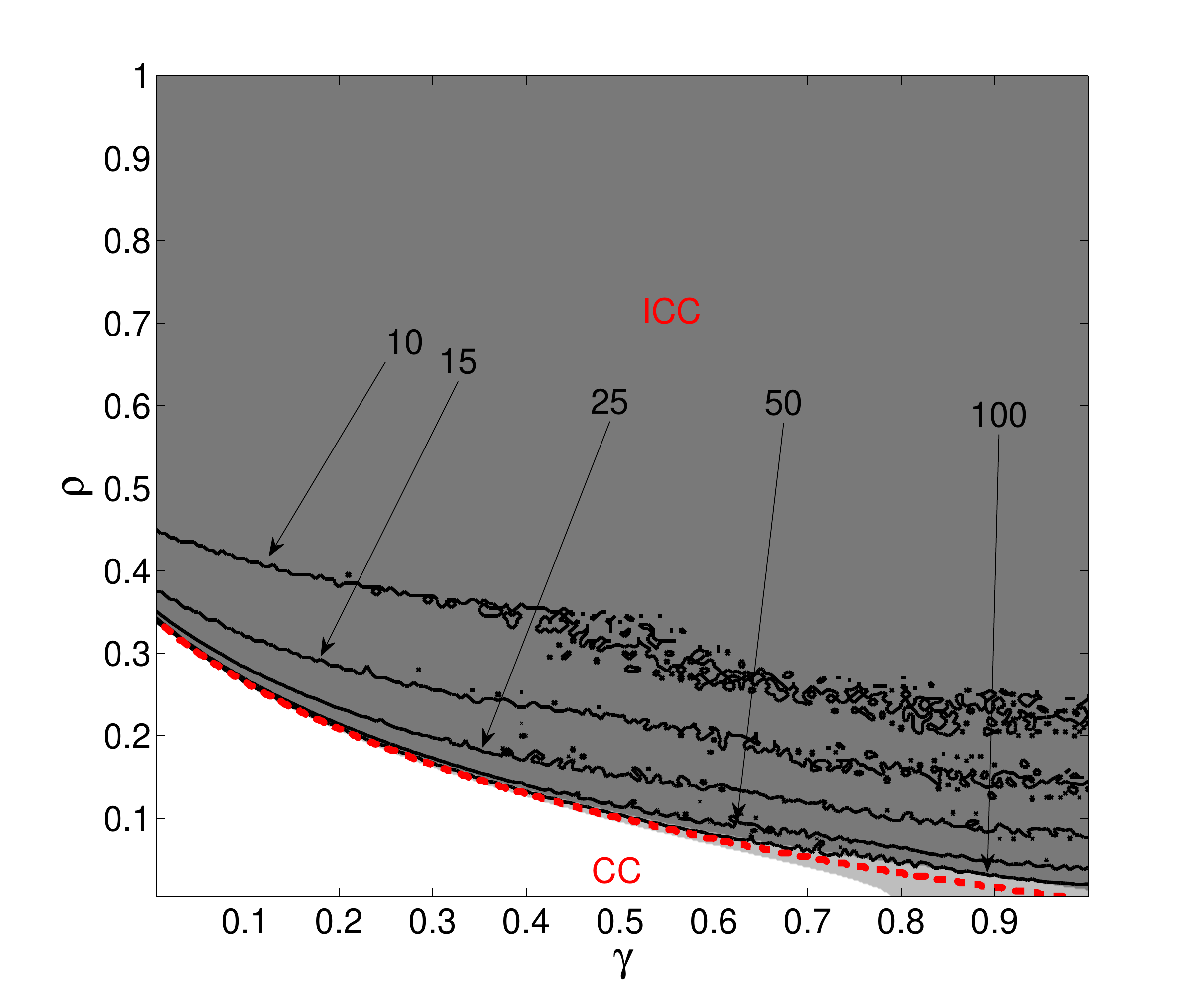}
        \caption{}
    \end{subfigure}
    
    \caption{Results in case of a triangle (a), a square (b), and a pentagon (c). The numbered contour lines show the simulation results. The shaded areas denote where the cone finding algorithm has found such cone (white), proven that no such cone exists (dark gray) or required more iterations (light gray). The red dashed line stands for the change of real and positive dominant eigenvalues to non-real or non-positive, conjectured to divide ICC from CC}
    \label{fig:simulations}
\end{figure}

\subsection{Comparison with direct simulation}

The results of the cone-finding algorithm have been compared with results of direct simulation of the linearized equations of motion. The same objects have been dropped with initial velocity $\vec{p}^{(0)}=[R_1,R_2,-1]^T$ and initial position $\vec{q}^{(0)}=[R_3,R_4,1+R_5]^T$ with $R_i$ being uniform random numbers over the interval $(0,0.1)$. The simulation was ended after 200 impacts or if \eqref{eq:vz<0} was violated. Fig. \ref{fig:simulations} shows several more or less irregular, solid curves, representing level curves of the total number of impacts during simulation. The object undergoes ICC in the region above the level curve of 200 impacts, and the rest of the parameter plane below this curve is our numerical approximation of the region where CC or PCC occurs. 

The level curve separating ICC from [CC or PCC] is fairly smooth. This result suggests that whether or not ICC occurs does not depend sensitively on the small perturbations $R_i$. At the same time, the level curves for lower numbers are quite irregular, which is an indication of sensistivity to our randomized initial conditions. Both findings are analogous to the results of \cite{goyal1,goyal2} for falling rods. 

CC was also separated from PCC in the simulations. We have seen that PCC leads to a state where the heights and the velocities of two adjacent vertices are 0. These states can be expressed as $\vec{q},\vec{p}=constant \cdot \vec{m}_{i,j}$ where $\vec{m}_{i,j}$ has been defined in Sec. \ref{sec:order}; $i,j$ are the indices of a pair of adjacent vertices, and the constant is positive for $\vec{q}$ and negative for $\vec{p}$. Thus, we detected a PCC in the simulation if the following criteria were met:
$$
\left | 1-\frac{\vec{q}^T\vec{m}_{i,j}}{|\vec{q}||\vec{m}_{i,j}|} \right | < \epsilon \ , \ \left |1+\frac{\vec{p}^T\vec{m}_{i,j}}{|\vec{p}||\vec{m}_{i,j}|}\right |<\epsilon
$$

PCC never occured in the simulation with $n=4,5$, which is consistent with Lemma \ref{lem:noPCC}. In the case of the triangle, PCC was found whenever the parameter values were on the left side of the solid curve marked as "PCC found" in Fig. \ref{fig:simulations}(a). This curve fits very well to the dashed curve given by Lemma \ref{lem:noPCC} (marked by the label "PCC possible" in the figure). We can draw the conclusions that the emergence of PCC is not sensitive to the randomized intial conditions, and the necessary condition of Lemma \ref{lem:noPCC} is probably exact. 

To compare the simulation results with the results of the cone-finding algorithms, we have added the background colours of Fig. \ref{fig:sajatertekek} to Fig. \ref{fig:simulations}. The figure strongly suggests that the effectively invariant cone exists whenever direct simulation indicates CC or PCC, i.e. that the conditions of Theorem \ref{thm:polygon2} are sharp. This is surprising, since the invariant cone approach focuses on velocity space and does not take into account how the positions of vertices in physical space evolve during motion. 

Finally we attempted to find a closed formula predicting the existence of CC dynamics and of an effectively invariant cone. Among others, we  examined the eigenvalues of matrices $\vec{U}_0\vec{P}_{2i\pi/n}$ ($i=0,2,...,n-1$) for many values of the parameters $\rho$ and $\gamma$. We found strong evidence that the transition between CC and ICC is linked to a qualitative change of the dominant eigenvalue for $i=1$. This surprising coincidence might be explained by the fact that all numerically simulated trajectories appear to become regular after an initial transient: an impact at a vertex is followed by an impact at its immediate neighbour, which corresponds to repeated application of the transformation \eqref{eq:ptransformation} with $i=1$. Proving that impact sequences converge to these regular patterns would require investigation of the full nonlinear dynamics in six dimensional state space (involving positions and velocities), which is beyond the scope of the present paper.    
 
As illustration, we show in the left panels of Fig.\ref{fig:sajatertekek} the sign of the dominant eigenvalue(s) of $\vec{U}_0\vec{P}_{2\pi/n}$ and whether they are real or complex. The inset of the figure for $n=3$ is a magnified detail. These results suggest that an effectively invariant cone exists if and only if the dominant eigenvalue is real and positive. The boundary of this region has also been added to Fig. \ref{fig:simulations} as a dashed curve. 

These findings enable us to formulate the following conjectures:
\begin{conjecture}
The following three statements are equivalent:\\
(i) The object $\mB$ undergoes CC or PCC for appropriately chosen initial conditions.\\
(ii) The matrices  $\vec{U}_0\vec{P}_{2i\pi/n}$ and cones $\mathcal{C}_i$ have an effectively invariant cone.\\
(iii) $\vec{U}_0\vec{P}_{2\pi/n}$ has a real and positive dominant eigenvalue.
\label{con:1}
\end{conjecture}

\begin{conjecture}
 The object $\mB$ undergoes PCC for appropriately chosen initial conditions
if and only if all eigenvalues of $\vec{U}_0\vec{P}_{-2\pi/n}\vec{U}_0\vec{P}_{2\pi/n}$are real.
\end{conjecture}

The conjecture predicts that a homogeneous, flat square plate with appropriate initial conditions may undergo CC if  the coefficient of restitution $\gamma$ is below $0.03$ and it always undergoes ICC otherwise. A homogeneous solid cube on the other hand has larger radius of gyration relative to its edge length, and it always undergoes ICC (even for $\gamma$ close to 0). The same conclusion holds for a dodecahedron. At the same time, a flat triangular plate, a regular tetrahedron and an octahedron may undergo PCC if $\gamma$ is below $0.01$, $0.025$ and $0.006$ or CC if $\gamma$ is below $0.04$, $0.116$, and $0.025$, respectively.  

In the following subsection, we will prove the equivalence of points (ii) and (iii) of Conjecture \ref{con:1} in the case of a square.

\subsection{A semi-analytic construction for squares}
\label{sec:square}

We begin the construction of an appropriate cone $\mathcal{K}$ with several steps of preparation. 

We will consider the central projection of velocity space to the plane $\mathcal{S}$ (as in Fig. \ref{fig:sajatertekek}). Several points will be identified in this plane, which are illustrated by Fig. \ref{fig:conedef}.

\begin{figure}
\begin{center}
\includegraphics[width=84mm]{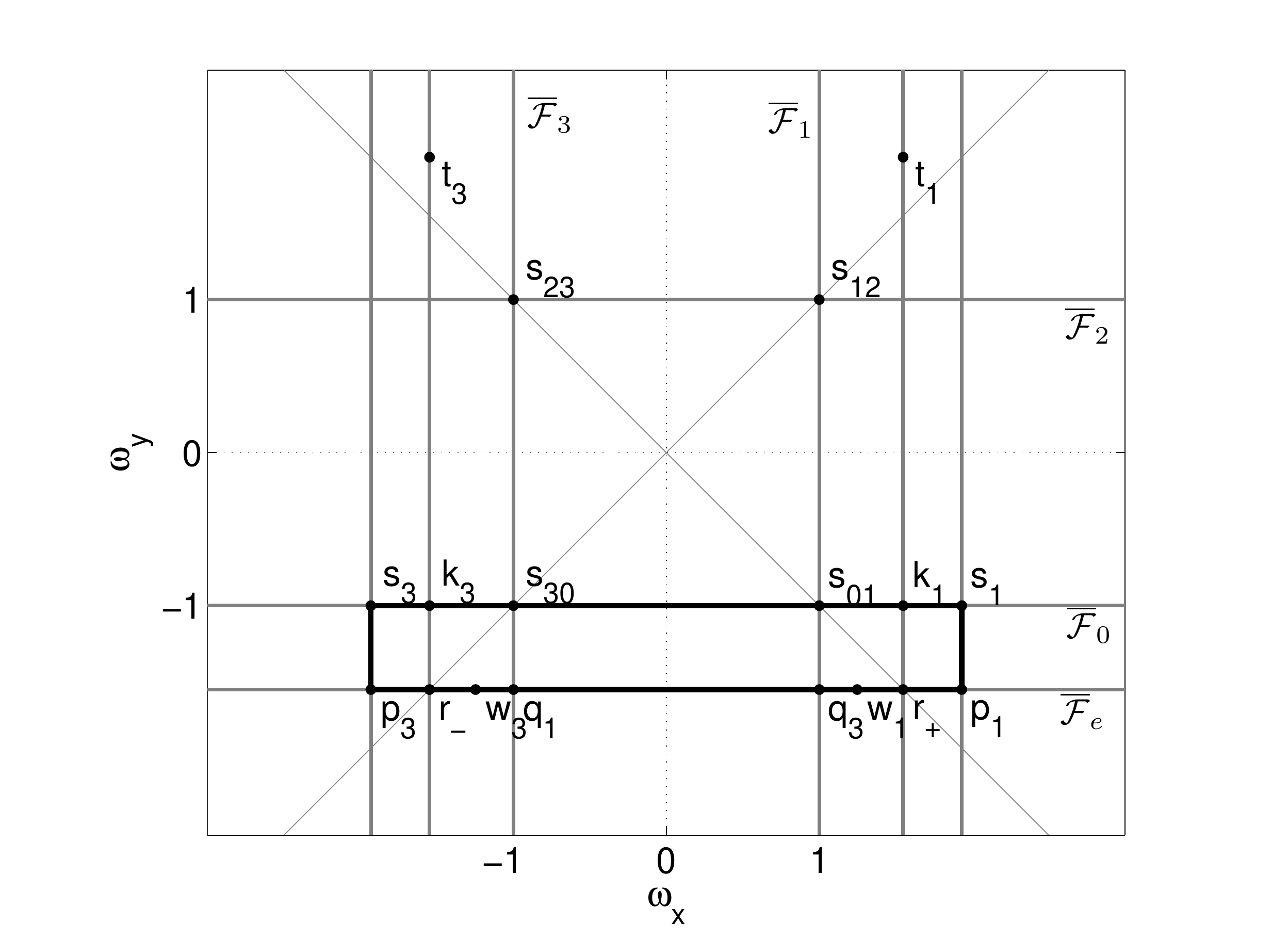}
\caption{Projection of $\mathcal{K}$ to the plane $\mathcal{S}$}
\label{fig:conedef}
\end{center}
\end{figure}

Recall that the planes $\mathcal{F}_i$ contain those points in $\mathbb{C}$ for which the height of one of the vertices is 0. There are four planes in $\mathbb{V}$, which are in the exact same positions as $\mathcal{F}_i$ in $\mathbb{C}$. These planes contain those points, for which the velocity of one of the points is zero. These planes project to $\mathcal{S}$ as four lines denoted by $\overline{\mathcal{F}_i}$ in Figure \ref{fig:conedef}. They enclose a square with vertices  $\vec{s}_{0,1}=[1,-1,-1]$, $\vec{s}_{1,2}=[1,1,-1]$, $\vec{s}_{2,3}=[-1,1,-1]$ and $\vec{s}_{3,0}=[-1,-1,-1]$. 

Assume that $\vec{U}_0\vec{P}_{\pi/2}$ has a unique, real and positive dominant eigenvalue $\lambda_{max}$. Let the corresponding eigenvector be $\vec{p}_{1}=[e_1,-e_2,-1]^T\in\mathcal{S}$ with $e_1,e_2\in\mR$. Our next goal is to find the approximate location of $\vec{p}_{1}$ within $\mathcal{S}$: 
\begin{lemma}\label{l:r2}
The coordinates of eigenvector ${\vec{p}}_{1}$ satisfy
$$e_1>e_2>1$$
\end{lemma}
\begin{proof} [Proof of Lemma \ref{l:r2}]
Impacts do not increase the kinetic energy of the object, from which it is easy to show that the dominant eigenvalue of $\vec{U}_0\vec{P}_{\pi/2}$ satisfies $\lambda_{max} \leq 1$.
From the fact that $\vec{p}_{1}$ is an eigenvector, it is also easy to derive that
\begin{align}
e_1=\frac{(\gamma+1)}{\lambda^2 \rho^2 + \lambda^2 + \rho^2 - \gamma} \quad e_2=\lambda_{max} e_1.
\end{align}

from which the $e_1>e_2$ relation immediately follows. The other relation ($e_2>1$) was verified numerically.

\end{proof}

Lemma \ref{l:r2} means that $\vec{p}_{1}$ is below $\overline{\mathcal{F}}_0$ and on the right side of the $\omega_x=\omega_y$ line in Fig. \ref{fig:conedef}. We now proceed by locating the dominant eigenvector of  $\vec{U}_0\vec{P}_{3\pi/2}$. The impact map $\vec{U}_0$ has a trivial invariance property:
$$
\vec{U}_0\vec{T}\vec{p}=\vec{TU}_0\vec{p}
\quad
\vec{T}=
\left[
\begin{matrix}
-1&0&0\\
0&1&0\\
0&0&1
\end{matrix}
\right]
$$
and it is also trivial that $\vec{P}_{3\pi/2}=\vec{TP}_{\pi/2}\vec{T}$. These two relations imply that the maps $\vec{U}_0\vec{P}_{\pi/2}$ and $\vec{U}_0\vec{P}_{3\pi/2}$ are related as 
\begin{align}
\vec{U}_0\vec{P}_{3\pi/2}=\vec{TU}_0\vec{P}_{\pi/2}\vec{T}
\label{eq:tuk}
\end{align}
Hence, the dominant eigenvalue and eigenvector of  $\vec{U}_0 \vec{P}_{3\pi/2}$ are $\lambda_{max}$ and $\vec{p}_3=\vec{T}\vec{p}_1$ as illustrated in Fig. \ref{fig:conedef}.

We also introduce the following new notations: 
 \begin{align}
 \vec{r}_+&=[e_2,-e_2,-1]^T\\
 \vec{q}_1&=[-1,-e_2,-1]^T\\
 \vec{k}_1&=[e2,-1,-1]^T\\
 \vec{s}_1&=[e_1,-1,-1]^T\\
 \vec{w}_1&=[\lambda_{max}^{-1},-e_2,-1]^T\\
 \vec{t}_1&=[e_2,e_1,-1]^T
 \end{align}
as well as
\begin{align}
\vec{r}_-=\vec{T}\vec{r}_+\\
\vec{q}_3=\vec{T}\vec{q}_1\\
\vec{k}_3=\vec{T}\vec{k}_1\\
\vec{s}_3=\vec{T}\vec{s}_1\\
\vec{w}_3=\vec{T}\vec{w}_1\\
\vec{t}_3=\vec{T}\vec{t}_1
\end{align}
 
According to Lemma \ref{l:r2}, $\vec{s}_3$, $\vec{k}_3$, $\vec{s}_{30}$, $\vec{s}_{01}$, $\vec{k}_1$, $\vec{s}_1$ lie along the line $\overline{\mathcal{F}_0}$ in the order of the list from left to right. At the same time, 
  $\vec{p}_3$, $\vec{r}_-$, $\vec{q}_1$, $\vec{q}_3$, $\vec{r}_+$, $\vec{p}_1$ lie along a line $\overline{\mathcal{F}}_e$ parallel to $\overline{\mathcal{F}_0}$ in the order of the previous list. The points $\vec{w}_3$,  $\vec{w}_1$ are also on $\overline{\mathcal{F}}_e$, and $\vec{w}_3$ is between $\vec{p}_3$ and $\vec{q}_1$ whereas $\vec{w}_1$ is between  $\vec{p}_1$ and $\vec{q}_3$. (See Fig. \ref{fig:conedef}.)

So far, we have defined all special points of $\mathcal{S}$, which will play a role in the upcoming construction. Now we present two lemmas on how the map $\vec{U}_0\vec{P}_{\pi/2}$ transforms the points defined above. First,

\begin{lemma}\label{l:r1}
The images of $\vec{s}_{30}$ and $\vec{q}_3$ under the map $\vec{U}_0\vec{P}_{\pi/2}$ are 
\begin{align} 
\vec{U}_0\vec{P}_{\pi/2}\cdot \vec{s}_{30}=\vec{s}_{01}\\
\vec{U}_0\vec{P}_{\pi/2}\cdot \vec{q}_1=\vec{k}_1
\end{align} 
\end{lemma}

\begin{proof} $\vec{P}_{\pi/2}$ represents a rotation by angle $\pi/2$, which implies 
$\vec{P}_{\pi/2}\cdot \vec{s}_{30}=\vec{s}_{01}$ and $\vec{P}_{\pi/2}\cdot \vec{q}_1=\vec{k}_1$. Furthermore, $\vec{s}_{01},\vec{k}_1\in \mathcal{F}_0$. For all $\vec{p}\in \mathcal{F}_0$, the velocity of vertex 0 is 0. Hence the impact map $\vec{U}_0$ leaves such values of $\vec{p}$ unchanged, which implies the statement of the lemma.
\end{proof}

\begin{lemma}\label{l:r4}
 The images of $\vec{p}_1$ and $\vec{s}_1$ under $\vec{U}_0\vec{P}_{\pi/2}$ are
\begin{align} 
\vec{U}_0\vec{P}_{\pi/2}\cdot \vec{p}_1=\lambda_{max} \vec{p}_1\\
\vec{U}_0\vec{P}_{\pi/2}\cdot \vec{s}_1=\lambda_{max} \vec{w}_1
\end{align}

\end{lemma}

\begin{proof}

The first statement is the immediate consequence of the fact that $\vec{p}_1$ is an eigenvector of $\vec{U}_0\vec{P}_{\pi/2}$.
In order to prove the second statement, we decompose $s_1$ as
$$
\vec{s}_1=\vec{p}_1+\vec{s}_{30}-\vec{q}_1
$$
from which
\begin{align}
\begin{split}
\vec{U}_0\vec{P}_{\pi/2}\vec{s}_1&=\lambda_{max} \vec{p}_1+\vec{s}_{01}-\vec{k}_1 \\
&=[\lambda_{max} e_1+1-e_2, \ \lambda_{max} e_2, \ -\lambda_{max} ]^T\\
&=[1, \ \lambda_{max} e_2, \ -\lambda_{max} ]^T\\
&=\lambda \vec{w}_1
\end{split}
\end{align}

\end{proof}

Now we are ready to formulate and prove an important result of this section: 

\begin{lemma}
If $n=4$ and the dominant eigenvalue  of $\vec{U}_0\vec{P}_{\pi/2}$ is positive and real, then the cone generated by the points $\vec{p}_1$, $\vec{p}_3$, $\vec{s}_1$, $\vec{s}_3$ satisfies conditions 1 and 2 of Theorem \ref{thm:polygon2}.
\label{lem:square}
\end{lemma} 

\begin{proof}[Proof of Lemma \ref{lem:square}:]
Equation \eqref{eq:vz<0} is satisfied by all 4 generating vectors, and thus by every point in the cone. Hence, condition 1 of Theorem \ref{thm:polygon2} is satisfied.

Condition 2 of the theorem requires effective invariance with respect to 3 maps and 3 cones. Below, we discuss each map one by one.

\begin{itemize}
\item \textbf{Map $\vec{U}_0\vec{P}_{\pi/2}$:} the corresponding condition is $\mathcal{C}_3$ (since $\vec{P}_{\pi/2}=\vec{P}_{-3\pi/2}$ ). The projection of the cone $\mathcal{C}_3$ to $\mathcal{S}$ is shown in Fig. \ref{fig:conditions} (f), from which  $\mK\cap\mathcal{C}_3$ is the cone generated by points $\vec{s}_{30}$, $\vec{q}_1$, $\vec{p}_1$ and $\vec{s}_1$ (horizontally hatched rectangle in Fig. \ref{fig:coneim} a)). According to Lemma \ref{l:r1} and Lemma \ref{l:r4}, the images of all these generating vectors (and thus the image of the entire cone $\mK\cap\mathcal{C}_3$)  are inside $\mK$ (vertically hatched rectangle in Fig. \ref{fig:coneim} a)).

\begin{figure}
\centering
    \begin{subfigure}{0.5\textwidth}
        \centering
        \includegraphics[width=84mm]{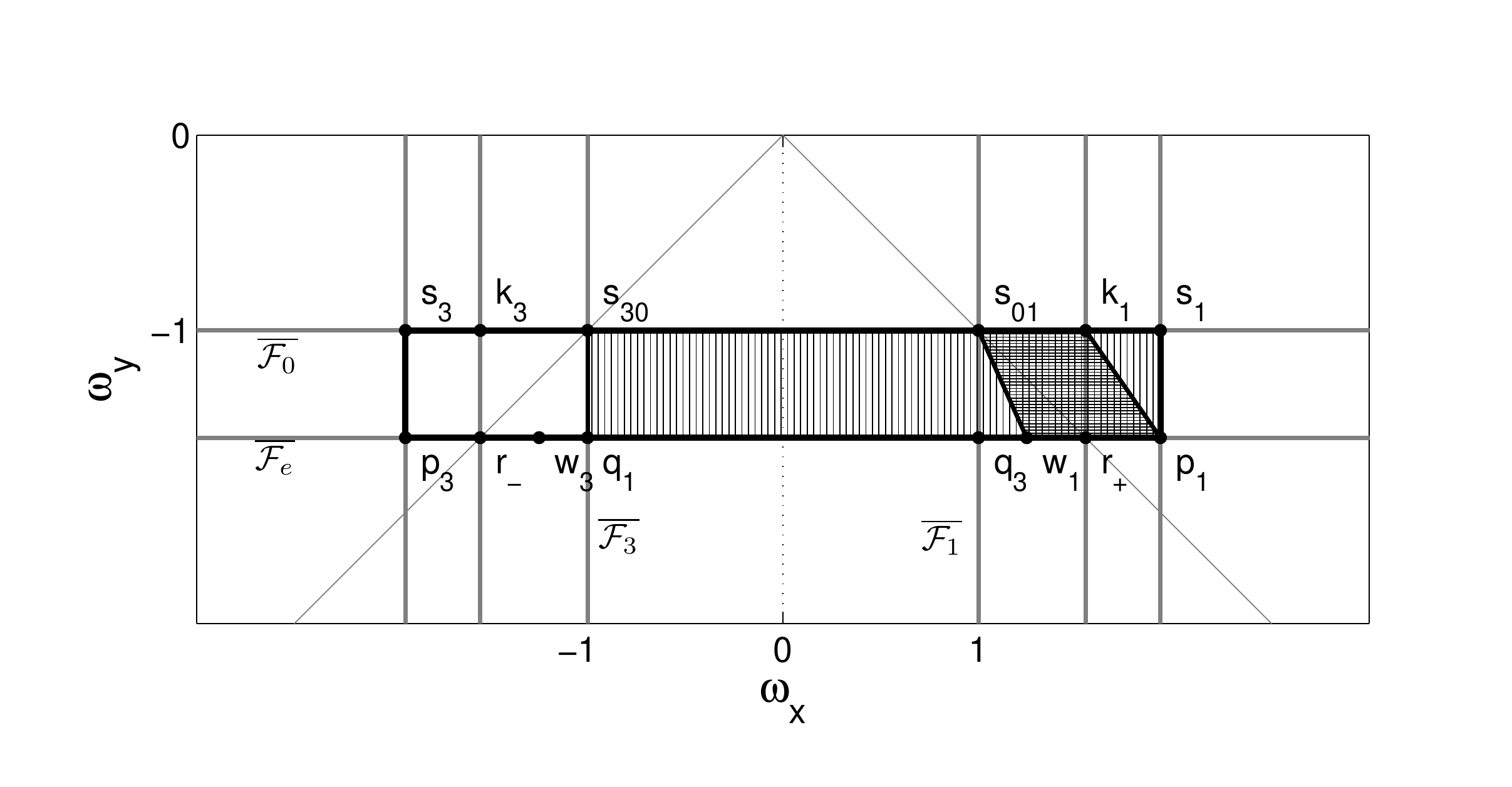}
        \caption{}
    \end{subfigure}%
    \\
        \centering
    \begin{subfigure}{0.5\textwidth}
        \centering
        \includegraphics[width=84mm]{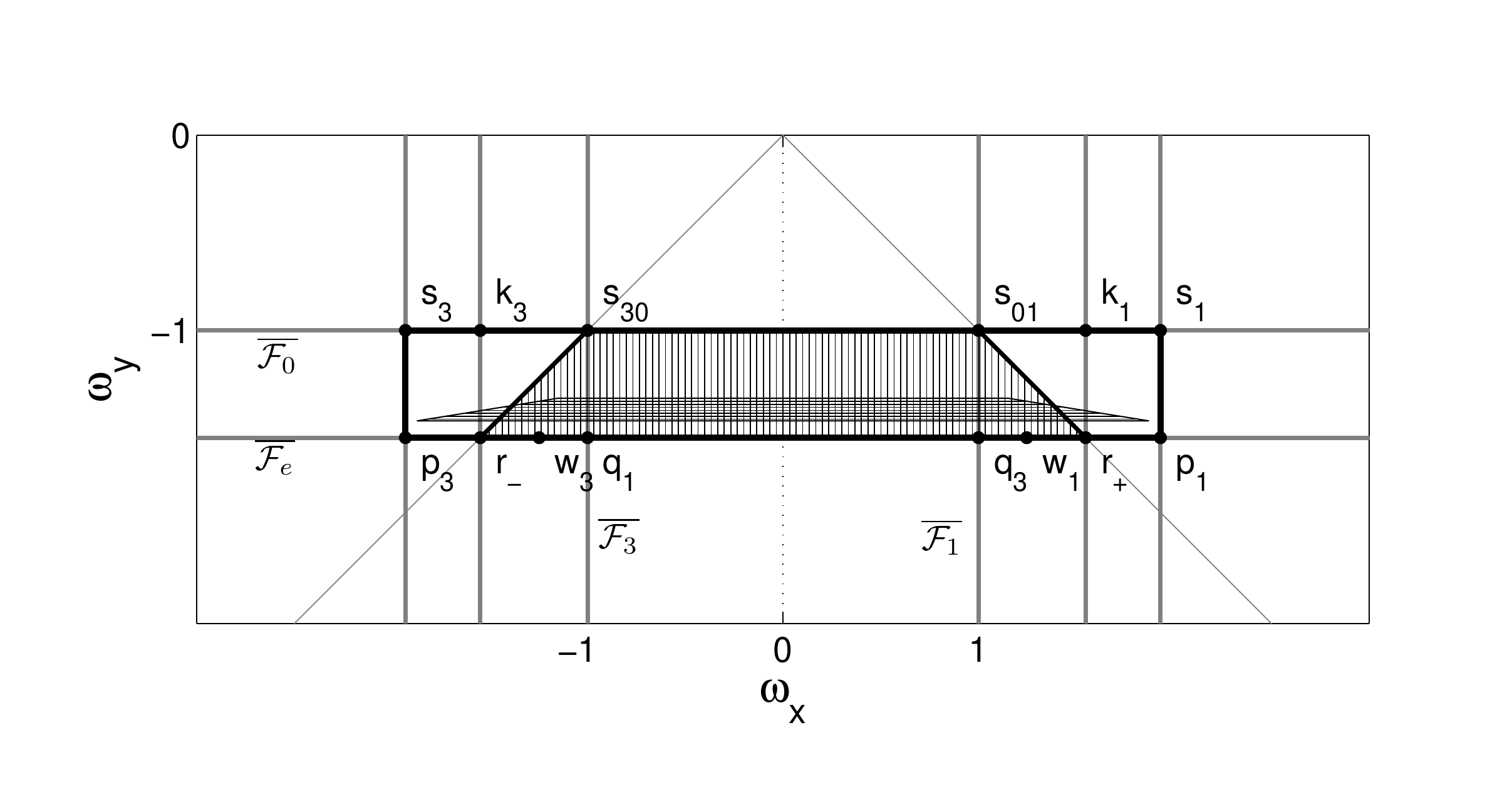}
        \caption{}
    \end{subfigure}%
\caption{a):Central projection of $\mathcal{K}$ (empty rectangle) of $\mathcal{K}\cap \mathcal{C}_3$ (vertical hatching) and of $\vec{U}_0\vec{P}_{\pi/2}(\mathcal{K}\cap \mathcal{C}_3)$ (horizontal hatching) to the plane  $\mathcal{S}$. b): projection of $\mathcal{K}\cap \mathcal{C}_2$ (vertical hatching) and of $\vec{U}_0\vec{P}_{\pi}(\mathcal{K}\cap \mathcal{C}_2)$ (horizontal hatching) to the plane the plane $\mathcal{S}$}
\label{fig:coneim}
\end{figure}%

\item \textbf{Map $\vec{U}_0\vec{P}_{3\pi/2}$:} due to the symmetry relation \eqref{eq:tuk}, and the invariance of cone $\mK$ to the transformation $T$, the proof in this case is the same as in the previous one.

\item \textbf{Map $\vec{U}_0\vec{P}_{\pi}$:} The projection of the cone $\mathcal{C}_2$ to $\mathcal{S}$ is shown in Fig. \ref{fig:conditions} (d), from which $\mK \cap \mathcal{C}_2$ is a cone generated by points $\vec{s}_{30}$, $\vec{r}_-$, $\vec{r}_+$ and $\vec{s}_{01}$. We can write
\begin{align}
\vec{U}_0\vec{P}_{\pi}\vec{s}_{30}=\vec{U}_0\vec{P}_{\pi/2}\vec{P}_{\pi/2}\vec{s}_{30}=\vec{U}_0\vec{P}_{\pi/2}\vec{s}_{01}\\
\vec{U}_0\vec{P}_{\pi}\vec{r}_-=\vec{U}_0\vec{P}_{\pi/2}\vec{P}_{\pi/2}\vec{r}_-=\vec{U}_0\vec{P}_{\pi/2}\vec{r}_+
\end{align}
The points $\vec{s}_{01}$ and $\vec{r}_+$ are in $\mK\cap\mathcal{C}_3$, hence their images under the map $\vec{U}_0\vec{P}_{\pi/2}$ are in $\mK$ (see first part of proof). Hence, $\vec{U}_0\vec{P}_{\pi}\vec{s}_{30},\vec{U}_0\vec{P}_{\pi}\vec{r}_-\in\mK$
It can be proven in an analogous way that $\vec{U}_0\vec{P}_{\pi}\vec{s}_{01},\vec{U}_0\vec{P}_{\pi}\vec{r}_+\in\mK$, and thus the image of $\mK\cap\mathcal{C}_2$ is in $\mK$, completing the proof. The situation described above is illustrated by Fig. \ref{fig:coneim} b).

\end{itemize}
\end{proof}

Whether or not Condition 3 of Theorem \ref{thm:polygon2} is satisfied, depends on the initial velocity of the object. Below we formulate a sufficient condition of this scenario:
\begin{lemma}
If all vertices are moving downwards initially,  then $\vec{p}^{(1)}\in\mK$ and thus condition 3 of Theorem \ref{thm:polygon2} is fulfilled. \label{l:k0}
\end{lemma}
\begin{proof}
Depending on the index of the vertex, which hits $\mP$ first, we have 
\begin{align}
(\vec{U}_0\vec{P}_{i\pi/2})^{-1}\vec{p}^{(1)}
\end{align}
with $i\in\{1,2,3,4\}$,
thus it is enough to show that the cone
\begin{align}
\mathcal{K}_0:=\cap_{i=0}^{3}\left(\vec{U}_0 \vec{P}_{i \pi/2}\right)^{-1}(\mathcal{K})
\end{align}
includes all points of $\mathbb{V}$ for which every vertex moves downwards. In the expression above, $\left(\vec{U}_0 \vec{P}_{i \pi/2}\right)^{-1}(\mathcal{K})$ is a shorthand notation for the transformed  image of $\mathcal{K}$ under the map $\left(\vec{U}_0 \vec{P}_{i \pi/2}\right)^{-1}$.

First, let us investigate the image $\vec{U}_0^{-1}(\mathcal{K})$. Since $\overline{\mathcal{F}}_0$ is the line corresponding to zero velocity of vertex $0$, we have $\vec{U}_0^{-1}\vec{s}_1=\vec{s}_1$ and $\vec{U}_0^{-1}\vec{s}_3=\vec{s}_3$. 
The eigenvector property of $\vec{p}_1$ means that $\left(\vec{U}_0 \vec{P}_{pi/2}\right)^{-1}\vec{p}_1=\lambda_{max}^{-1} \vec{p}_1$, which can be rearranged as
$$
\vec{U}_0^{-1}\vec{p}_1=\lambda_{max}^{-1}\vec{P}_{\pi/2}\vec{p}_1
$$
Hence the projection of $\vec{U}_0^{-1}\vec{p}_1$ to $\mathcal{S}$ is $\vec{t}_1$. Similarly, the projection of $\vec{U}_0^{-1}\vec{p}_3$ is $\vec{t}_3$. In sum, $\vec{U}_0^{-1}(\mathcal{K})$ is the cone spanned by $\vec{s}_1$, $\vec{s}_3$, $\vec{t}_1$, and $\vec{t}_3$. The related cones $(\vec{U}_0\vec{P}_{i\pi/2})^{-1}(\mK)$ ($i=1,2,3$) can be obtained simply by rotating the cone $\vec{U}_0^{-1}(\mathcal{K})$ with an angle of $i\pi/2$ (Fig. \ref{fig:coneimk0}). The intersection of the resulting four cones is the cone generated by the points $\vec{s}_{i,j}$. This cone contains  exactly those points for which every vertex of the square moves downwards, which completes our proof. \end{proof}

\begin{figure}[h]
\begin{center}
\includegraphics[width=84mm]{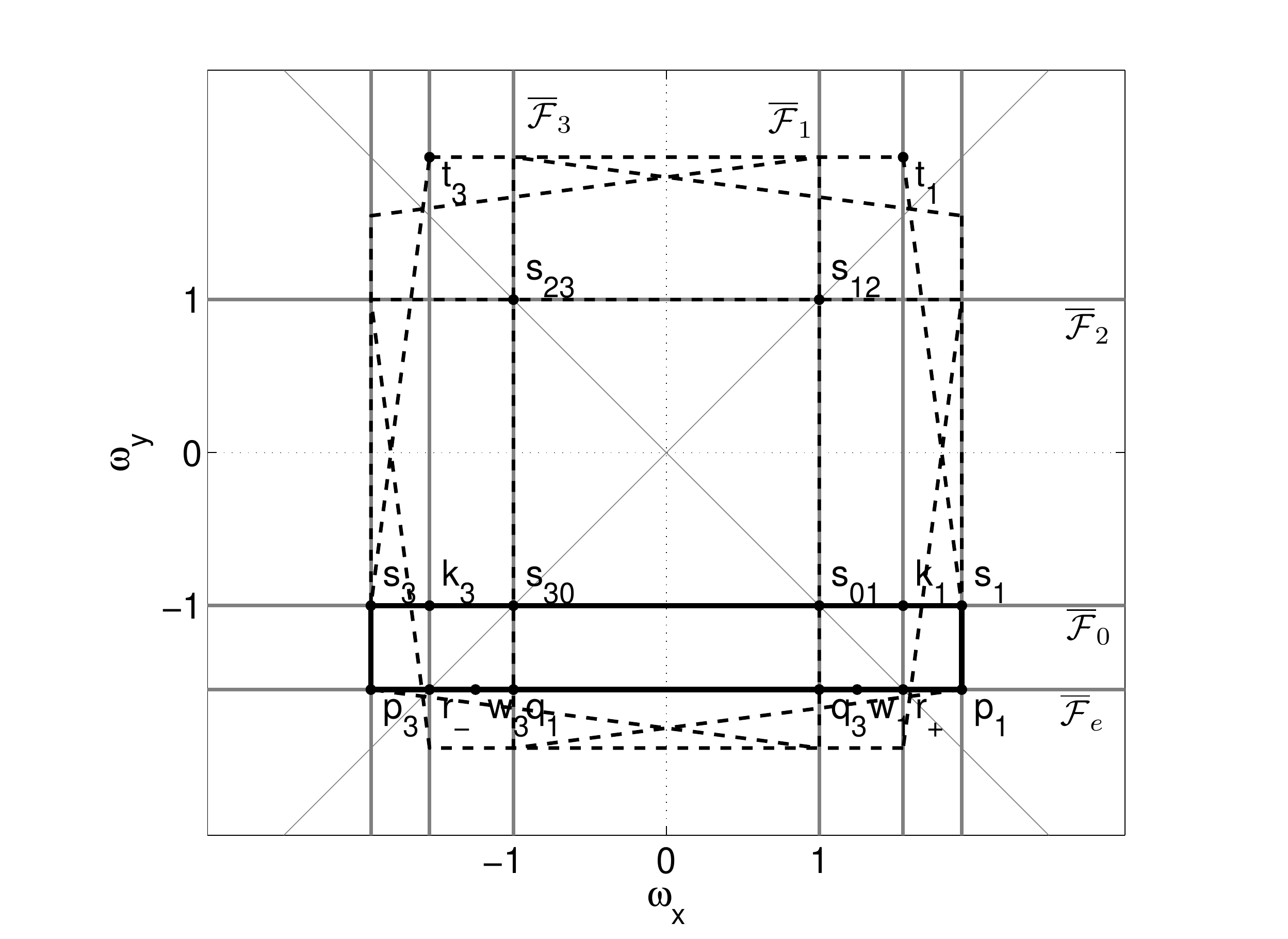}
\caption{Images of $\mathcal{K}$ under the four maps $\left(\vec{U}_0\vec{P}_{i\pi/4}\right)^{-1}$ with $i\in \{0,1,2,3\}$}.
\label{fig:coneimk0}
\end{center}
\end{figure}

Our last task is to summarize the results achieved so far. This is done in
\begin{theorem}
If all the vertices of a square move towards $\mathcal{S}$ initially, and the dominant eigenvalue of the map $\vec{U}_0\vec{P}_{\pi/2}$ is positive and real, then the square undergoes CC.
\label{thm:squarefinal}
\end{theorem}

\begin{proof}

According to Lemma \ref{l:k0}, we have $\vec{p}^{(1)}\in\mathcal{K}$, and Lemma \ref{lem:square} implies that conditions 1 and 2 of Theorem \ref{thm:polygon2} are also satisfied. Hence, the object must undergo CC or PCC by Theorem \ref{thm:polygon2}. At the same time, Lemma \ref{lem:noPCC} outrules PCC, which completes the proof. 
\end{proof}

\subsection{Affine images of regular $n$-gons}
\label{sec:affine}

So far, we have been dealing with objects whose contact points form a regular $n$-gon, with emphasis on squares. Nevertheless, in the case of flat objects (i.e. $z_*=0$ in \eqref{eq:koord}), these results can also be applied to contact point arrangements, which are affine images of the $n$-gon. This class of arrangements includes arbitrary rectangles, parallelograms or triangles. 

In order to show this, we will compare the motion of the previously examined object $\mB$ with vertices $\vec{r}_i$ ($i=0,1,...,n$) and the motion of an affine image $\mB^*$ of this object, which is obtained by the following transformation:
$$
\vec{r} \rightarrow \vec{S}\vec{r},\;\vec{r}\in\mR^3
$$
with $\vec{S}$ being an invertible matrix of form:
\begin{align}
\vec{S}=  \begin{bmatrix}
    a & b & 0 \\
    c & d & 0 \\
    0 & 0 & 1
  \end{bmatrix}
  \;\;a,b,c,d\in\mR
\end{align}
The vertices of the new object are
\begin{align}
\vec{r}_i^*=\vec{S}\vec{r}_i 
\label{eq:rip}
\end{align}
We will assume that the mass density distributions $\delta(\vec{r})$ and $\delta^*(\vec{r})$ are also related as
\begin{align}
\delta^*(\vec{S}\vec{r})=\chi\cdot\delta(\vec{r})
\label{eq:deltarel}
\end{align}
for some $\chi\in\mR$. The coefficients of restitution are assumed to be the same for the two objects. In what follows, physical quantities associated with the transformed object will be denoted by an asterisk. 

In this subsection, we focus on flat objects, for which the $z$ coordinate of every point is close to 0. In this case, the generalized mass moment of inertia matrices of $\mB$ and $\mB^*$ are given by the following volume integrals:
\begin{align}
\begin{split}
\vec{\Theta}&=\int \delta(\vec{r})\vec{f}_i \vec{f}_i^T \ dV=\\
&=\int \delta(\vec{r})\vec{P}_{-\pi/2}(\vec{r}^l+\vec{u}_3)(\vec{r}^l+\vec{u}_3)^T\vec{P}_{\pi/2} \ dV
\end{split}
\\
\begin{split}
\vec{\Theta}^*&=\int \delta^*(\vec{r}) \vec{f}_i \vec{f}_i^T \ dV=\\
&=\int \delta^*(\vec{r})\vec{P}_{-\pi/2}(\vec{r}^l+\vec{u}_3)(\vec{r}^l+\vec{u}_3)^T\vec{P}_{\pi/2}\ dV
\end{split}
\end{align}
We can use \eqref{eq:deltarel} and the identity $\vec{S}\vec{u}_3=\vec{u}_3$ to establish the relation 
\begin{align}
\vec{\Theta}^*=\chi \ det(\vec{S})\cdot  \vec{S}_p\Theta \vec{S}_p^T
\label{eq:thetap}
\end{align}
where
\begin{align}
\vec{S}_p=\vec{P}_{-\pi/2}\vec{SP}_{\pi/2}
\end{align}

Let $\vec{q}$ and $\vec{p}$ denote the initial generalized coordinates and velocities of $\mB$, and let
\begin{align}
\vec{q}^*&=\vec{S}_p^{-T}\vec{q}
\label{eq:qrelation}\\
\vec{p}^*&=\vec{S}_p^{-T}\vec{p}
\label{eq:prelation}
\end{align} 
be the initial position and velocity of $\mB^*$. Then, the heights of vertex $i$ of the two objects can be determined with the aid of \eqref{eq:fi}. Since the objects are flat ($z_*=0$ in \eqref{eq:fi}), we have 
$$
h_i=\vec{f}_i^T\vec{q}=\left( \vec{P}_{-\pi/2} (\vec{r}_i+\vec{u}_3) \right)^T \vec{q} 
$$
in the case of $\mB$, and the exact same values for $\mB^*$ because:
\begin{align}
\begin{split}
h_i^*&=\left( \vec{P}_{-\pi/2} (\vec{r}_i^*+\vec{u}_3) \right)^T \cdot \vec{q}^*\\
&=\left( \vec{P}_{-\pi/2}\vec{S}(\vec{r}_i+\vec{u}_3)\right)^T \cdot \left(\vec{P}_{-\pi/2}\vec{SP}_{+\pi/2}\right)^{-T} \vec{q}\\
&=\left( \vec{P}_{-\pi/2} (\vec{r}_i+\vec{u}_3) \right)^T \cdot \vec{q}
\end{split}
\end{align}
Similarly, the normal velocities of vertices $i$ of the two objects are also equal. Hence, they  hit the ground at the same time and with the same vertex.

The post-collision velocities are determined by the impact maps:  $\vec{p}^+=\vec{U}_i\vec{p}$ and $\vec{p}^{*+}=\vec{U}_i^*\vec{p}^*$, respectively. Combining \eqref{eq:impactequ} with the relations \eqref{eq:thetap} and \eqref{eq:rip} yields
\begin{align}
\vec{U}_i^*=\vec{S}_p^{-T}\vec{U}_i\vec{S}_p^{T}.
\end{align}
It follows then that the post-impact velocity of $\mB^*$ becomes
\begin{align}
\begin{split}
\vec{p}^{*+}& = \vec{U}_i^*\vec{p}^* \\ 
 & = \vec{S}_p^{-T}\vec{U}_i \vec{S}_p^{T}\cdot \vec{S}_p^{-T}\vec{p}\\
   & = \vec{S}_p^{-T}\cdot \vec{U}_i\vec{p}\\
   & = \vec{S}_p^{-T}\vec{p}^+
   \end{split}
\end{align}
i.e. the relation \eqref{eq:prelation} is preserved by the impact maps. Hence we conclude that 
  
\begin{theorem}
If a flat object $\mB^*$ is an affine image of another flat object $\mB$ in the sense of \eqref{eq:rip} and \eqref{eq:deltarel}, futhermore the initial conditions of the two objects are related according to \eqref{eq:qrelation} and \eqref{eq:prelation}, then \eqref{eq:qrelation} and \eqref{eq:prelation} are preserved during the entire motion and thus $\mB^*$ undergoes CC if and only if $\mB$ does so.
\label{thm:affine}
\end{theorem}

For example, a flat, homogeneous, rectangular plate is the affine image of a square plate. Theorem \ref{thm:squarefinal} and Theorem \ref{thm:affine} implies that both objects undergo CC if the initial velocities satisfy condition 3 of Theorem \ref{thm:squarefinal}, furthermore the coefficient of restitution $\gamma$ is below $0.03$. In a similar fashion, our numerical results suggest that a flat, homogeneous triangular plate of arbitrary shape undergoes CC if $\gamma<0.04$ and PCC if $\gamma<0.01$.

\section{Summary}\label{sec:sum}

In this paper, we have examined the chattering motion of three-dimensional objects with more than 2 potential contact points and rotational symmetry when hitting an immobile surface. The motion of the system was examined in velocity space (i.e. a 3D projection of the 6D state space), where it was modelled by a non-deterministic, discrete dynamical system. We have applied invariant cone theory, its recent generalization (common invariant cones) as well  as a novel generalization (effectively invariant cones) to find sufficient conditions of complete chatter, i.e. an infinite sequence of impacts driving the system to a complete halt.

We have developed a numerical algorithm to verify the sufficient condition for regular polygon-shaped arrangement of the contact points, as well as a semi-analytical verification in the case of a square. The dynamics of the system has also been examined via direct numerical simulation, which suggests that our sufficient conditions are indeed exact, moreover whether or not CC occurs can be predicted by solving a simple matrix eigenvalue problem as in the case of slender rods.

The paper leaves several open questions, including proofs of the conjectures drawn from numerical results, and the exact conditions of CC with respect to the initial velocity of the object. Our future plans additionally include several broader extensions of these results, including a detailed investigation of the case of inelastic impacts as well as the analysis of objects whose contact points are in irregular positions.

The results of the paper have many potential applications. The investigations of the rod problem by \cite{goyal1,goyal2} were motivated partially by the need to understand the motion of objects dropped to the floor in order to improve the shock protection of electronic devices. Needless to say, since these devices are three-dimensional blocks rather than slender rods, our new results represent a significant improvement in this direction.  

Chattering is tightly related to rocking block problems \cite{Housner}, which have been in the focus of interest for several decades mainly due to their role in earthquake design. The three-dimensional motion of rocking blocks was not investigated until recently \cite{konstantinidis2007dynamics,di2014comparison}. We believe that conditions of CC in three dimensions will help engineers in improving the earthquake-resistance of free-standing block-like structures (such as pillars of bridges).

Another delicate situation where chattering-type behaviour occurs is the docking of a spacecraft at another, or the landing of a spacecraft with multiple legs on a solid surface without active control. If the system is modelled as a rigid body and the landing takes place in a microgravitational environment, CC corresponds to successful landing whereas ICC means that the spacecraft either topples or leaves the landing site. A somewhat similar scenario has been realized during the recent Rosetta mission of the European Space Agency, when the three-legged lander unit Philae failed to anchor itself to the surface of comet C67-G, and tumbled above the comet surface for several hours. Identifying the final location of the lander required two years of active search by the mission team \cite{philaefound}.
 
 \appendix
\section{Appendix: proof of Lemma \ref{lem:cone}}

The four vectors listed  in Lemma \ref{lem:cone} belong to the cone $\mathcal{C}_j$, because they all satisfy \eqref{eq:vz<0}, furthermore it is easy to find a pair of points $\vec{q}_0 \in F_0 \cap \mathcal{F}$ and $\vec{q}_j \in F_j \cap \mathcal{F}$ such that $\vec{q}_j-\vec{q}_0$ is equal to any of these four vectors. Thus it suffices to prove that any vector satisfying the conditions of the lemma is inside $\mathcal{C}_j$.

Because of the conditions $\vec{q}_0 \in F_0 \cap \mathcal{F}$ and $\vec{q}_j \in F_j \cap \mathcal{F}$, we can write

\begin{align}
\vec{q}_0&=\alpha_1 \vec{m}_{0,1}+\alpha_{-1} \vec{m}_{n-1,0} 
\\
\vec{q}_j&=\beta_1 \vec{m}_{j,j+1}+\beta_{-1} \vec{m}_{j-1,j} 
\end{align}
with $\alpha_1,\alpha_{-1},\beta_1,\beta_{-1}\geq 0$, yielding
\begin{align}
\begin{split}
\vec{q}_j-\vec{q}_0&=\beta_1 \vec{m}_{j,j+1}+\beta_{-1} \vec{m}_{j,j-1}-\\
&-\alpha_1 \vec{m}_{0,1} -\alpha_{-1} \vec{m}_{n-1,0}\\
&= \beta_1 (\vec{m}_{j,j+1}-\vec{m}_{0,1})+\\
&+\beta_{-1}(\vec{m}_{j,j-1}-\vec{m}_{n-1,0})-\\
&-(\alpha_1-\beta_1) \vec{m}_{0,1} -(\alpha_{-1}-\beta_{-1}) \vec{m}_{n-1,0}
\end{split}\label{eq:conic}
\end{align}
This is a linear combination of the four vectors, which appear in the statement of the lemma, but the third and fourth coefficients are not necessarily positive. The constraint \eqref{eq:vz<0} can be expressed as
\begin{align}
\alpha_1+\alpha_{-1}\geq \beta_1+\beta_{-1}
\label{eq:alphabeta}
\end{align}
which implies that at most one of the coefficients $(\alpha_1-\beta_1)$ and $(\alpha_{-1}-\beta_{-1})$ in \eqref{eq:conic} is negative. Hence we have 3 possibilities with respect to the signs of these coefficients: $(\alpha_{-1}-\beta_{-1}),(\alpha_1-\beta_1)\geq 0$ or $(\alpha_{-1}-\beta_{-1}) \leq 0\leq (\alpha_1-\beta_1)$ or $(\alpha_{-1}-\beta_{-1})\geq 0\geq (\alpha_1-\beta_1)$.

\textbf{First case:} if $(\alpha_{-1}-\beta_{-1}) $ and $ (\alpha_1-\beta_1) \geq 0$; then (\ref{eq:conic}) implies that $\vec{q}_j-\vec{q}_0 \in\mathcal{C}_{j}$ and thus all vectors $\vec{p}$ satisfying \eqref{eq:i>>j} are also in $\mathcal{C}_{j}$. 

\textbf{Second case:} if $(\alpha_{-1}-\beta_{-1})\leq 0$, then we express $\vec{m}_{n-1,0}$ as
\begin{align}
\vec{m}_{n-1,0}=\gamma_{-1} \vec{m}_{j-1,j} +\gamma_1 \vec{m}_{j,j+1}+ \gamma_0 \vec{m}_{0,1} \label{eq:4thvertex}
\end{align}
where 
\begin{align}
\gamma_{-1}+\gamma_1+\gamma_{0}=1 \quad \vert \ \gamma_{-1}\leq 0;\ \gamma_{0},\gamma_{1}\geq 0
\label{eq:gammak}
\end{align}
We can rearrange \eqref{eq:4thvertex} as
\begin{align}
\begin{split}
\vec{m}_{n-1,0}&= \gamma_{-1} (\vec{m}_{j-1,j}-\vec{m}_{n-1,0})+\gamma_{-1} \vec{m}_{n-1,0} + \\
& +\gamma_1(\vec{m}_{j,j+1}-\vec{m}_{0,1})+(\gamma_1+\gamma_0) \vec{m}_{0,1}\\
&=\frac{\gamma_{-1}}{1-\gamma_{-1}}(\vec{m}_{j-1,j}-\vec{m}_{n-1,0})+\\
&+\frac{\gamma_1+\gamma_0}{1-\gamma_{-1}} \vec{m}_{0,1} +\\
& +\frac{\gamma_1}{1-\gamma_{-1}} (\vec{m}_{j,j+1}-\vec{m}_{0,1})
\end{split} 
\end{align}
which is then plugged into
(\ref{eq:conic}) to obtain $\vec{q}_j-\vec{q}_0$ as a conical combination of three vectors:
\begin{align*}
\vec{q}_j-\vec{q}_0&=\left(\underbrace{\beta_1}_{\geq 0}-
\underbrace{(\alpha_{-1}-\beta_{-1})}_{\leq 0}\frac{\overbrace{\gamma_1}^{\geq 0}}{1-\underbrace{\gamma_{-1}}_{\leq 0}}\right)(\vec{m}_{j,j+1}-\vec{m}_{0,1})+
\\
&+\left(\underbrace{\beta_{-1}}_{\geq 0}\frac{1}{1-\underbrace{\gamma_{-1}}_{\leq 0}}-\underbrace{\alpha_{-1}}_{\geq 0}\frac{\overbrace{\gamma_{-1}}^{\leq 0}}{1-\underbrace{\gamma_{-1}}_{\leq 0}}\right)(\vec{m}_{j,j-1}-\vec{m}_{n-1,0})+
\\
&+\left(\overbrace{(\alpha_1-\beta_1)+(\alpha_{-1}-\beta_{-1})}^{\geq 0 \ by \ \eqref{eq:alphabeta}}\underbrace{\frac{\gamma_1+\gamma_0}{1-\gamma_{-1}}}_{=1  \ by \ \eqref{eq:gammak}} \right) (-\vec{m}_{0,1})
\end{align*}

All three coefficients are positive, so we conclude that $\vec{q}_j-\vec{q}_0\in \mathcal{C}_{j}$ and thus all vectors $\vec{p}$ satisfying \eqref{eq:i>>j} are also in $\mathcal{C}_{j}$.

\textbf{Third case:} the proof is completely analogous to the second case.

\begin{acknowledgements}
This work has been supported by the National Research, Development, and Innovation Office of Hungary under grant 104501.
\end{acknowledgements}

\bibliographystyle{ieeetr}

\bibliography{2016_clattering}

\end{document}